\documentclass[10pt,conference]{IEEEtran}

\usepackage{ellipsis, ragged2e} 
\usepackage[l2tabu, orthodox]{nag} 
\usepackage{amsmath,amssymb,amsthm}
\usepackage{varioref}
\usepackage{hyperref}
\usepackage{xcolor}

\usepackage[english]{babel}
\usepackage[T1]{fontenc}
\usepackage[utf8]{inputenc}
\usepackage{csquotes}
\usepackage{booktabs}
\usepackage{thm-restate}

\usepackage{framed}
\usepackage{pgfplots}
\pgfplotsset{compat=1.16}
\usepackage{tikz}
\usetikzlibrary{shapes,positioning,arrows,arrows.meta,calc,automata,matrix,fit,backgrounds}
\tikzstyle{state}+=[minimum size = 6mm, inner sep=0,outer sep=1]

\colorlet{disabled}{lightgray}
\tikzstyle{state}=[draw,rectangle,inner sep=5pt,rounded corners=2pt]
\tikzstyle{action}=[font=\small,inner sep=0pt,outer sep=3pt]
\tikzstyle{actionnode}=[circle,draw=black,fill=black,minimum size=1mm,inner sep=0,outer sep=0]
\tikzstyle{actionedge}=[draw,-]
\tikzstyle{prob}=[font=\scriptsize,inner sep=0pt,outer sep=1pt]
\tikzstyle{probedge}=[draw,->]
\tikzstyle{directedge}=[draw,->]
\tikzset{chainarrow/.tip={Stealth[length=3pt]}}
\tikzset{>=chainarrow}

\usepackage{xparse}
\usepackage{mathtools}
\usepackage{environ}
\usepackage{marvosym}
\usepackage{bbm}

\usepackage{url}
\usepackage{algorithmicx,algorithm}
\usepackage[noend]{algpseudocode}

\usepackage{booktabs}
\usepackage{multirow}

\usepackage{enumitem}

\usepackage[capitalize]{cleveref}

\newtheorem{theorem}{Theorem}[section]
\newtheorem{corollary}[theorem]{Corollary}
\newtheorem{lemma}[theorem]{Lemma}
\newtheorem{proposition}[theorem]{Proposition}
\theoremstyle{definition}
\newtheorem{definition}[theorem]{Definition}
\newtheorem{remark}[theorem]{Remark}
\newtheorem{example}[theorem]{Example}
\newtheorem{assumption}{Assumption}

\DeclarePairedDelimiter{\delimabs}{\lvert}{\rvert}
\DeclarePairedDelimiter{\delimcardinality}{\lvert}{\rvert}
\DeclarePairedDelimiter{\delimnorm}{\lVert}{\rVert}

\NewDocumentCommand{\abs}{sm}{\IfBooleanTF{#1}{\delimabs*{#2}}{\delimabs{#2}}}
\NewDocumentCommand{\cardinality}{sm}{\IfBooleanTF{#1}{\delimcardinality*{#2}}{\delimcardinality{#2}}}
\NewDocumentCommand{\norm}{sm}{\IfBooleanTF{#1}{\delimnorm*{#2}}{\delimnorm{#2}}}
\NewDocumentCommand{\powerset}{r()}{2^{#1}}

\newcommand{\unionSym}{\cup}
\newcommand{\unionBin}{\mathbin{\unionSym}}

\newcommand{\intersectionSym}{\cap}
\newcommand{\intersectionBin}{\mathbin{\intersectionSym}}
\newcommand{\UnionSym}{\bigcup}

\newcommand{\union}{\unionBin}

\newcommand{\intersection}{\intersectionBin}
\newcommand{\Union}{\UnionSym}

\newcommand{\Rationals}{\mathbb{Q}}
\newcommand{\Reals}{\mathbb{R}}

\newcommand{\RealsNonneg}{\Reals_{\geq 0}}
\newcommand{\RealsPositive}{\Reals_{> 0}}

\DeclareMathOperator{\support}{supp}

\NewDocumentCommand{\Distributions}{d()}{\IfNoValueTF{#1}{\mathcal{D}}{\mathcal{D}(#1)}}
\NewDocumentCommand{\Measures}{d()}{\IfNoValueTF{#1}{\Pi}{\Pi(#1)}}
\NewDocumentCommand{\integral}{d<> m m}{\IfNoValueTF{#1}{\int #2\,d#3}{\int_{#1} #2\,d#3}}
\NewDocumentCommand{\Expectation}{s d[]}{\IfNoValueTF{#2}{\mathbb{E}}{\mathbb{E}\IfBooleanTF{#1}{\left[#2\right]}{[#2]}}}
\NewDocumentCommand{\Probability}{s d[]}{\mathop{\mathrm{Pr}}\IfValueT{#2}{\IfBooleanTF{#1}{\left[#2\right]}{[#2]}}}

\newcommand{\Cyl}{\mathit{Cyl}}

\newcommand{\sched}{\mathfrak{S}}
\newcommand{\tsched}{\mathfrak{T}}

\newcommand{\MC}{\mathsf{M}}
\newcommand{\MDP}{\mathcal{M}}

\newcommand{\States}{S}

\newcommand{\initialstate}{{\hat{s}}}
\newcommand{\Actions}{A}
\NewDocumentCommand{\stateactions}{r()}{{\Actions}(#1)}
\NewDocumentCommand{\mctransitions}{d()}{\IfNoValueTF{#1}{\delta}{\delta(#1)}}
\NewDocumentCommand{\mdptransitions}{d()}{\IfNoValueTF{#1}{\Delta}{\Delta(#1)}}
\NewDocumentCommand{\sgtransitions}{d()}{\IfNoValueTF{#1}{\Delta}{\Delta(#1)}}
\NewDocumentCommand{\reward}{d()}{\IfNoValueTF{#1}{{r}}{{r}(#1)}}

\DeclareMathOperator{\CylOp}{Cyl}
\NewDocumentCommand{\cylinder}{d()}{\IfNoValueTF{#1}{{\CylOp}}{{\CylOp}(#1)}}

\newcommand{\infinitepath}{\rho}
\newcommand{\finitepath}{\varrho}
\NewDocumentCommand{\Infinitepaths}{d<>}{\IfNoValueTF{#1}{\mathsf{Paths}}{\mathsf{Paths}_{#1}}}
\NewDocumentCommand{\Finitepaths}{d<>}{\IfNoValueTF{#1}{\mathsf{FPaths}}{\mathsf{FPaths}_{#1}}}

\NewDocumentCommand{\FinitepathsFromTo}{d<> r() r()}{\Finitepaths<#1>^{{#2}\to{#3}}}

\newcommand{\strategy}{\pi}

\NewDocumentCommand{\Strategies}{d<>}{\IfNoValueTF{#1}{\Pi}{\Pi_{#1}}}
\NewDocumentCommand{\StrategiesM}{d<>}{\IfNoValueTF{#1}{\Pi}{\Pi_{#1}}^{\mathsf{M}}}
\NewDocumentCommand{\StrategiesMD}{d<>}{\IfNoValueTF{#1}{\Pi}{\Pi_{#1}}^{\mathsf{MD}}}
\newcommand{\last}[1]{last(#1)}

\DeclareMathOperator{\SccsOp}{SCC}
\DeclareMathOperator{\BsccsOp}{BSCC}
\DeclareMathOperator{\EcsOp}{EC}
\DeclareMathOperator{\MecsOp}{MEC}

\NewDocumentCommand{\Sccs}{r()}{\SccsOp(#1)}
\NewDocumentCommand{\Bsccs}{r()}{\BsccsOp(#1)}
\NewDocumentCommand{\Ecs}{d()}{\IfNoValueTF{#1}{\EcsOp}{\EcsOp(#1)}}
\NewDocumentCommand{\Mecs}{d()}{\IfNoValueTF{#1}{\MecsOp}{\MecsOp(#1)}}

\NewDocumentCommand{\ProbabilityMC}{s r<> d[]}{\mathsf{Pr}_{#2}\IfNoValueF{#3}{\IfBooleanTF{#1}{\!\left[#3\right]\!}{[#3]}}}
\NewDocumentCommand{\ProbabilityMDP}{s r<> r<> d[]}{\mathsf{Pr}_{#2}^{#3}\IfNoValueF{#4}{\IfBooleanTF{#1}{\!\left[#4\right]\!}{[#4]}}}
\NewDocumentCommand{\ProbabilitySG}{s r<> r<> d[]}{\mathsf{Pr}_{#2}^{#3}\IfNoValueF{#4}{\IfBooleanTF{#1}{\!\left[#4\right]\!}{[#4]}}}
\NewDocumentCommand{\ProbabilityMDPmax}{s r<> d[]}{\mathsf{Pr}_{#2}^{\max}\IfNoValueF{#3}{\IfBooleanTF{#1}{\!\left[#3\right]\!}{[#3]}}}
\NewDocumentCommand{\ProbabilityMDPsup}{s r<> d[]}{\mathsf{Pr}_{#2}^{\sup}\IfNoValueF{#3}{\IfBooleanTF{#1}{\!\left[#3\right]\!}{[#3]}}}

\NewDocumentCommand{\ExpectationMC}{s r<> d[]}{\mathbb{E}_{#2}\IfValueT{#3}{\IfBooleanTF{#1}{\!\left[#3\right]\!}{[#3]}}}
\NewDocumentCommand{\ExpectationMDP}{s r<> r<> d[]}{\mathbb{E}_{#2}^{#3}\IfValueT{#4}{\IfBooleanTF{#1}{\!\left[#4\right]\!}{[#4]}}}
\NewDocumentCommand{\ExpectationSG} {s r<> r<> d[]}{\mathbb{E}_{#2}^{#3}\IfValueT{#4}{\IfBooleanTF{#1}{\!\left[#4\right]\!}{[#4]}}}

\newcommand{\reach}{\lozenge}

\newcommand{\MPsup}{\overline{\mathit{MP}}}
\newcommand{\MPinf}{\underline{\mathit{MP}}}
\newcommand{\MP}{\mathit{MP}}

\newcommand{\rew}{r}

\newcommand{\cE}{E}

\DeclareMathOperator{\totalreward}{TR}
\DeclareMathOperator{\multreward}{MR}

\DeclareMathOperator{\ERisk}{ERisk}

\NewDocumentCommand{\ERiskSG}{r<> r<> r[]}{\ERisk_{#1}^{#2}(#3)}

\NewDocumentCommand{\pathreward}{}{{\overline{\multreward}_{\reward}}}
\NewDocumentCommand{\pathrewardinf}{}{{\underline{\multreward}_{\reward}}}

\NewDocumentCommand{\Pathpartitions}{m d()}{\IfValueTF{#2}{\text{Part}_{#1}^{#2}}{\text{Part}_{#1}}}

\usepackage[final]{todonotes}
\setuptodonotes{inline}
\newlist{enumerateindent}{enumerate}{1}
\setlist[enumerateindent]{label=\arabic*),leftmargin=*}

\begin{document}

\title{ Multiplicative Rewards in Markovian Models  }


\author{anonymous}

\author{\IEEEauthorblockN{Christel Baier}
\IEEEauthorblockA{TU Dresden }
\and
\IEEEauthorblockN{Krishnendu Chatterjee}
\IEEEauthorblockA{IST Austria }
\and
\IEEEauthorblockN{Tobias Meggendorfer}
\IEEEauthorblockA{Lancaster University Leipzig}
\and
\IEEEauthorblockN{Jakob Piribauer}
\IEEEauthorblockA{TU Dresden, 
Leipzig University}
}

\maketitle

\begin{abstract}
This paper studies the expected value of \emph{multiplicative rewards}, where rewards obtained in each step are \emph{multiplied} (instead of the usual addition), in Markov chains (MCs) and Markov decision processes (MDPs).
One of the key differences to additive rewards is that the expected value may diverge to $\infty$ not only  due to recurrent, but also due to transient states.

For MCs, computing the value  
is shown to be possible in polynomial time given an oracle for the comparison of succinctly represented integers (CSRI), which is only known to be solvable in polynomial time subject to number-theoretic conjectures.
Interestingly, distinguishing whether the value is $\infty$ or $0$ is at least as hard as CSRI, while determining if it is one of these two can be done in polynomial time.
In MDPs, the optimal value can be computed in polynomial space.
Further refined complexity results and results on the complexity of optimal schedulers are presented.
The techniques developed for MDPs additionally allow to solve the multiplicative variant of the stochastic shortest path problem.
Finally, for MCs and MDPs where an absorbing state is reached almost surely, all considered problems are solvable in polynomial time.

\end{abstract}

\section*{Acknowledgments}
{This work was partly funded by the ERC CoG 863818 (ForM-SMArt), the Austrian Science Fund (FWF) 10.55776/COE12, the DFG Grant 389792660 as part of TRR 248 (Foundations of Perspicuous Software Systems), the Cluster of Excellence EXC 2050/1 (CeTI, project ID 390696704, as part of Germany’s Excellence Strategy), and by the BMBF (Federal Ministry of Education and Research) in DAAD project 57616814 (SECAI, School of Embedded and Composite AI) as part of the program Konrad Zuse Schools of Excellence in Artificial Intelligence..}

\section{Introduction}

Markov chains are standard models for the analysis and verification of systems exhibiting \emph{probabilistic} behavior.  Markov decision processes (MDPs) extend Markov chains with \emph{non-determinism} and are widely used in various fields of science such as computer science \cite{DBLP:books/wi/Puterman94}, biology \cite{paulsson2004summing}, epidemiology \cite{G_mez_2010}, and chemistry \cite{gillespie1976general}, to name a few.
In computer science, MDPs play a crucial role in formal verification \cite{DBLP:books/daglib/0020348},  control theory, and artificial intelligence, e.g., as the underlying model for reinforcement learning \cite{kaelbling1996reinforcement}, among others.
In a nutshell, an MDP works as follows:
In each state of an MDP, a so-called \emph{scheduler} chooses from a set of available actions. Afterwards, the successor state is chosen randomly according to a probability distribution associated with the chosen action.
Once a scheduler is fixed, the process is purely probabilistic.

\paragraph*{Rewards}
To model quantitative aspects of a system such as runtime, energy consumption, or utility, Markov chains and MDPs can be equipped with reward structures.
In MDPs, the central arising problem is  finding a scheduler optimizing  the expected value of the obtained reward.
The most prominent reward mechanisms are accumulated/total (additive) rewards (rewards  obtained in each step are summed up), discounted accumulated rewards (rewards obtained after $n$ steps are discounted with factor $\lambda^n$ for $\lambda\in (0,1)$), and mean payoff (average reward accumulated per step).
Optimizing the expected values of these quantities has extensively been studied for decades (see, e.g.,~\cite{DBLP:books/wi/Puterman94}) and is doable in polynomial time.

\paragraph*{Multiplicative rewards}
This paper investigates \emph{multiplicative rewards} in Markov chains and MDPs. Instead of adding up rewards obtained in each step, the rewards are simply multiplied.
Typical situations in which rewards behave multiplicatively arise whenever the change of some quantity is proportional to the amount already obtained.
This is a common theme in  chemical and biological systems where, e.g., the current rate of a reaction or the change of the size of a population is proportional to the current concentration of a substance or the current population size. Also in financial settings, multiplicative rewards naturally occur. The return of an investment depends on the interest rate as well as the current size of the investment.

\begin{example}
Suppose a plant population decreases or increases yearly depending on the climatic conditions during the year. Favorable climatic conditions lead to a $30\%$ increase of the population while adverse conditions lead to a $25\%$ decrease. Further suppose both conditions are equally likely each year. We can model this with a simple Markov chain with multiplicative rewards as shown in \Cref{fig:plants}.
Starting with good climate, in each step there is a $0.5$ probability to move to the good state and $0.5$ probability to move to the bad state.
Along the resulting path, the rewards $1.3$ and $0.75$ (denoted in the states) are obtained and multiplied together.
In this case, the expected multiplicative reward over an infinite time horizon is $0$ due to the negative logarithmic mean payoff of
$
	0.5\cdot \log(1.3) + 0.5\cdot \log(0.75)= 0.5 \cdot \log(0.975)<0
$.
In particular, this means the plant population will become extinct.
\end{example}

\begin{figure}
	\centering
	\begin{tikzpicture}[auto,initial text=,xscale=0.9]
		\node[state,initial=left] at (0, -.6) (s1) {good, $1.3$};
		\node[actionnode] at (2, -.2) (a11) {};
		\node[state] at (5, -.6) (s2) {bad, $0.75$};
		\node[actionnode] at (3, -1) (b11) {};

		\path[->,probedge]
			(a11) edge[swap,bend left=10]  node[prob] {0.5} (s1)
			(a11) edge node[prob] {0.5} (s2)
			(b11) edge node[prob] {0.5} (s1)
			(b11) edge[swap,bend left=10] node[prob] {0.5} (s2)
			;
		\path[-,actionedge]
			(s1) edge [bend left]  (a11)
			(s2) edge [bend left] (b11)
		;
	\end{tikzpicture}
	\caption{
		A Markov chain  modelling a plant population.
	}
	\label{fig:plants}
\end{figure}
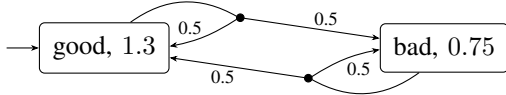

\paragraph*{Exponential utility functions}
A further situation in which multiplicative rewards naturally occur is the use of  \emph{exponential utility functions}, a type of utility function frequently used in economics and suitable to express \emph{risk-aversion}.
Applied to the maximization of additive rewards in MDPs, this works as follows. If the additive reward along an execution is $w$, the associated utility is $f(w)=1-2^{-\lambda w}$ for a fixed parameter $\lambda>0$.
This monotone function ensures that high rewards lead to a high utility. But the change in utility decreases for higher rewards. Low outcomes, however, have a strong effect on the utility and should be avoided.
A scheduler maximizing the expected exponential utility hence has to ensure high values of the additive rewards while keeping the probability of low outcomes small. The use of the exponential utility function
incentivizes a risk-averse behavior.

An exponential utility function converts the additive reward mechanism into a multiplicative one: If the rewards $w_1,w_2,\dots$ are collected on a run of the MDP leading to an additive reward of $w=w_1+w_2+\dots$, the exponential utility $f(w)$ can be obtained as $1-f(w_1)\cdot f(w_2) \cdot \dots$. So, maximizing the expected utility is the same as minimizing the expected multiplicative rewards after transforming the reward $r(s)$ of a state $s$ to $f(r(s))$.
Hence, the optimization of multiplicative rewards studied in this paper can be used for the risk-averse optimization of additive rewards.

Such an exponential utility function also lies at the heart of the \emph{entropic risk measure}  \cite{follmer2002convex}.
This notion has been widely studied in finance and operation research (see, e.g.,~\cite{follmer2011entropic,brandtner2018entropic}).
For Markovian models, more precisely Markov chains, MDPs, as well as stochastic games extending MDPs with a second player, the optimization of the entropic risk measure applied to additive non-negative rewards has been studied in \cite{DBLP:journals/iandc/BaierCMP24}.

\paragraph*{Divergence in recurrent and transient states}
One noteworthy difference between additive and multiplicative rewards  is as follows.
In Markov chains, the expected value of accumulated additive rewards can only diverge to $\infty$ due to the rewards collected at \emph{recurrent} states. 
These are states that are visited infinitely often with positive probability. 
In contrast, the expected value of multiplicative rewards may additionally be infinite due to rewards collected at \emph{transient} states, i.e.\ states that are almost surely only visited finitely often.

\begin{example}
Consider the  Markov chain $\MC$ depicted in \Cref{fig:divergence}.
From the initial state $s_1$, the Markov chain loops back to $s_1$ or moves to the absorbing state $s_2$ with probability $1/2$ each.
The state $s_1$ is hence visited twice in expectation and almost surely only finitely often. As soon as $s_2$ is reached, reward $1$ is obtained in every step not influencing the multiplicative reward.
The expected value of the total multiplicative reward $\mathrm{MR}$ is $\ExpectationMC<\MC, s_1>[\mathrm{MR}] = {\sum}_{i=1}^\infty (1 / 2)^{i} \cdot 3^i =\infty$, as the path visiting $s_1$ $i$ times has probability $(1/2)^i$ for each $i\geq 1$.
 In the case of additive rewards,
the contribution of a state to the expected total accumulated reward is the expected number of visits to the state times the reward of the state and hence always finite for transient states.
\end{example}

\begin{figure}
	\centering
	\begin{tikzpicture}[auto,initial text=]
		\node[state,initial=left] at (0, -.6) (s1) {$s_1$, $3$};
		\node[actionnode] at (2, -.6) (a11) {};
		\node[state] at (4, -.6) (s2) {$s_2$, $1$};


		\path[->,directedge]
			(s2) edge[loop right] (s2)
	;
		\path[->,probedge]
			(a11) edge  [bend left]  node[prob] {0.5} (s1)
			(a11) edge node[swap,prob] {0.5} (s2)
		;
		\path[-,actionedge]
			(s1) edge  (a11)
		;
	\end{tikzpicture}
	\caption{
		The Markov chain $\MC$ with multiplicative rewards diverging in a transient state.
	}
	\label{fig:divergence}
\end{figure}
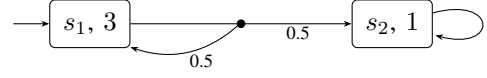

\paragraph*{Contrast to additive logarithmic rewards in MDPs}
Along a single path, the multiplicative reward and the sum of the logarithms of the rewards along the path behave, effectively (up to applying the exponential function to the result), the same.
Surprisingly, there is a difference between optimizing the expected multiplicative reward in MDPs to optimizing expected additive logarithmic rewards.
In fact, different schedulers are required in general to maximize the expected multiplicative reward and the expected additive logarithmic reward, respectively, as the following example shows.

\begin{example}
Consider the MDP $\MDP$ with reward function $\rew$ depicted in \Cref{fig:contrast_MDP}. 
A scheduler only can to choose in the initial state whether to move to $s_2$ using action $a$ or to move to $s_3$ with action $b$.
From $s_3$, the goal is reached directly, while from $s_2$, there is a probabilistic loop back to $s_2$.
When choosing $a$ to move to $s_2$, we get an expected additive reward with respect to the logarithmic rewards ${\log} \circ {\rew}$ of $\frac53\cdot \log(2) = \frac53$ and a multiplicative reward with respect to $\rew$ of
\[
	{\sum}_{i=1}^{\infty} 2^i \cdot \frac35 \cdot \big(\frac25\big)^{(i-1)} =\frac65 \cdot {\sum}_{j=0}^{\infty} \big(\frac45\big)^j = 6.
\]
Choosing $b$ to move to $s_3$ yields an additive  logarithmic reward of $\log(4) = 2$ and a multiplicative reward of $4$.
So, the additive logarithmic reward is maximized by action $b$, while the multiplicative reward is maximized by action $a$.
\end{example}

\begin{figure}
	\centering
	\begin{tikzpicture}[auto,initial text=]
		\node[state,initial=left] at (0, 0) (s1) {$s_1$, $1$};
		\node[state] at (2, .5) (s2) {$s_2$, $2$};
		\node[actionnode] at (3.5, .25) (a2) {};
		\node[state] at (2, -.5) (s3) {$s_3$, $4$};
		\node[state] at (5, 0) (s4) {$s_4$, $1$};

		\path[->,directedge]
			(s4) edge[loop right] (s4)
			(s1) edge node {$a$} (s2)
			(s1) edge node[swap] {$b$} (s3)
			(s3) edge (s4)
	;
		\path[->,probedge]
			(a2) edge  [bend right]  node[swap,prob] {2/5} (s2)
			(a2) edge node[prob] {3/5} (s4)
		;
		\path[-,actionedge]
			(s2) edge  (a2)
		;
	\end{tikzpicture}
	\caption{
		The MDP $\MDP$ illustrating the difference between optimizing multiplicative rewards and optimizing additive logarithmic rewards.
	}
	\label{fig:contrast_MDP}
\end{figure}
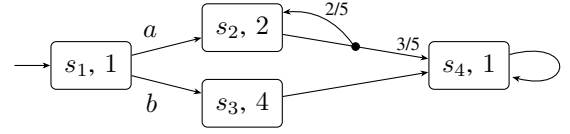

\subsection*{Contribution}
In this paper, we investigate the computation of the expected value of multiplicative rewards in Markov chains as well as the optimization of this expected  in MDPs.
For an infinite path $\infinitepath = \infinitepath_0 \infinitepath_1 \cdots$ in a Markov chain or MDP with states $\States$ and with reward function $\reward\colon \States\to \Reals_{>0}$, we define the (lower and upper) multiplicative rewards $\pathreward(\infinitepath)$ and  $\pathrewardinf(\infinitepath)$ as the lim sup and lim inf of $\prod_{k=0}^n \reward(\infinitepath_k)$ for $n\to \infty$.
In several cases, the problems to be solved in the process reduce to the comparison of succinctly represented  integers (CSRI).
Given two integers represented  as $\prod_{i=1}^k a_i^{b_i}$ with $k,a_1,\dots,a_k,b_1,\dots,b_k\in \mathbb{N}$, 
 CSRI asks to decide whether the first is larger than the second.
Potentially surprisingly, CSRI is not known to be solvable in polynomial time  \cite{DBLP:journals/toct/EtessamiSY14}, but lies in PSPACE.
Further, the problem is in P if one of the following two number-theoretic conjectures holds \cite{DBLP:journals/toct/EtessamiSY14}:
(1)~Baker's refinement \cite{Baker+1998+37+44} of the \emph{ABC conjecture} of  Masser and Oesterl\`e or
(2)~the \emph{Lang-Waldschmidt conjecture} \cite{lang1978elliptic}.

The results for Markov chains $\MC$ with rational reward function $\reward$ and initial state $\initialstate$ are as follows:
\begin{enumerateindent}
\item
The values $\ExpectationMC<\MC, \initialstate>[\pathreward] $ and $\ExpectationMC<\MC, \initialstate>[\pathrewardinf] $ can be computed in polynomial time using a CSRI-oracle.
Thus, the computation is possible in polynomial space. Further, if CSRI is in P -- so in particular if one of the two mentioned conjectures holds -- the computation is possible in polynomial time.

\item
Deciding whether the  values are finite or infinite is at least as hard as CSRI.

\item
If an absorbing state is reached almost surely in  $\MC$, then $\ExpectationMC<\MC, \initialstate>[\pathreward] = \ExpectationMC<\MC, \initialstate>[\pathrewardinf]$ and the value 
can be computed in polynomial time.
\end{enumerateindent}

For MDPs $\MDP$ with rational reward function and initial state $\initialstate$, we write $\ExpectationMDP<\MDP, \initialstate><\max>[\pathreward] $ for the supremum of the expected value of $\pathreward$ when ranging over all schedulers and likewise for $\pathrewardinf$. Our results show that this supremum is indeed a maximum.
We obtain the following results:
\begin{enumerateindent}[resume]
\item The values $\ExpectationMDP<\MDP, \initialstate><\max>[\pathreward]$  and  $\ExpectationMDP<\MDP, \initialstate><\max>[\pathrewardinf] $ can be computed in polynomial space.
\item If  an absorbing state is reached almost surely under any scheduler,  $\ExpectationMDP<\MDP, \initialstate><\max>[\pathreward] = \ExpectationMDP<\MDP, \initialstate><\max>[\pathrewardinf] $
and this value can be computed in polynomial time.
\item There is an MD-scheduler $\sched$ such that $\ExpectationMDP<\MDP, \initialstate><\max>[\pathreward] = \ExpectationMDP<\MDP, \initialstate><\sched>[\pathreward]$.
\item If  $\ExpectationMDP<\MDP, \initialstate><\max>[\pathrewardinf] < \infty$, there is an MD-scheduler $\sched$ such that $\ExpectationMDP<\MDP, \initialstate><\max>[\pathrewardinf] = \ExpectationMDP<\MDP, \initialstate><\sched>[\pathrewardinf]$.
\item If $\ExpectationMDP<\MDP, \initialstate><\max>[\pathrewardinf] = \infty$, an optimal scheduler $\sched$ may require infinite memory to obtain $\ExpectationMDP<\MDP, \initialstate><\sched>[\pathrewardinf] = \infty$. 
\item
The threshold problem, i.e., deciding $\ExpectationMDP<\MDP, \initialstate><\max>[\pathreward]\geq \vartheta$ (and $\ExpectationMDP<\MDP, \initialstate><\max>[\pathrewardinf]\geq \vartheta$) for a given $\theta \in \Rationals$, is in $\textrm{NP}^{\textrm{CSRI}}$ and so in NP if one of the two mentioned conjectures holds.
\end{enumerateindent}
Additionally, these results allow us to solve the \emph{multiplicative stochastic shortest path problem} in polynomial space. Analogously to the classical (additive) stochastic shortest path problem,
this problem asks for the optimal expected value of the multiplicative reward when only allowing schedulers that reach a given  absorbing target state, i.e.\ a state with only a self loop and reward $1$, almost surely.

\subsection*{Related Work}

As described above, exponential utility functions and the entropic risk measure applied to additive rewards implicitly require the treatment of multiplicative rewards.
In \cite{DBLP:journals/iandc/BaierCMP24}, the computation and optimization of the entropic risk measure in Markov chains, MDPs, and stochastic games with total additive rewards is studied.
Key differences to our work arise for the following two reasons:
Firstly, \cite{DBLP:journals/iandc/BaierCMP24} considers non-negative additive rewards that are turned into multiplicative rewards from the interval $[0,1]$ by the used exponential utility function.
In contrast, we allow arbitrary positive rewards not bounded by $1$.
Secondly, starting with rational additive rewards, the used utility function transforms rewards into algebraic or even transcendental numbers.
This leads to computational challenges addressed in \cite{DBLP:journals/iandc/BaierCMP24}.
In this paper, we consider rational multiplicative rewards for all algorithmic problems.
Further work on entropic risk and exponential utility functions in MDPs in the finite horizon setting \cite{howard1972risk}, applied to discounted rewards \cite{jaquette1976utility}, and on infinite state MDPs \cite{di1999risk} focused on optimality equations and general convergence results of value iteration, but did not address the \emph{algorithmic} challenges for finite-state MDPs.

Further, matrix multiplication games and the subclass of entropy games \cite{DBLP:conf/stacs/AsarinCDDHK16,DBLP:journals/mst/AkianGGG19,DBLP:conf/icalp/AllamigeonGKS22}
are a class of  games using a multiplicative reward mechanism: Two players alternatingly choose matrices from a finite set. One player wants to maximize the growth rate of the product of these matrices while the other player tries to minimize this rate. So, conceptually these games can be seen as using a ``multiplicative mean payoff'' as objective, in contrast to the multiplicative total reward we use.
Matrix multiplication games also allow for an  interpretation as stochastic mean-payoff games (see \cite{DBLP:journals/mst/AkianGGG19}).

In \cite{DBLP:journals/mor/Rothblum84} the limit behaviour of finite-horizon / $n$-step rewards (and the limit growth / decline rate thereof) is considered.
In our context, this relates to the expected multiplicative reward after $n$ steps (discussed as $\multreward^n_{\reward}$ in \cref{rem:infinitevszero}), which surprisingly is fundamentally different from our objective.
Secondly, they focus on a mathematical discussion (e.g.\ optimality criterion etc.) and not on algorithmic aspects, which is the focus of our work.

For the technical treatment of recurrent states, we partly rely on results from \cite{BBDGS2018}.
In \cite{BBDGS2018}, a characterization of \emph{end components} (strongly connected sub-MDPs) regarding the divergence behavior of additive rewards is provided.
In order to apply these results, we  apply the logarithm to our multiplicative rewards posing additional computational challenges.
Further, we adapt a construction from \cite{BBDGS2018} to remove end components without diverging behavior.
For the identification of these components, however, the results of \cite{BBDGS2018} are not applicable in our setting and we develop new techniques to this end.

\section{Preliminaries}
\label{sec:prelim}

A \emph{Markov decision process} (MDP) is a tuple $\mathcal{M} = (S,A,\mdptransitions)$
where $S$ is a finite set of states,
$A$ a finite set of actions, and
$\mdptransitions \colon S \times A \times S \to [0,1]$  the
transition probability function.
We require that
$\sum_{t\in S}\mdptransitions(s,\alpha,t) \in \{0,1\}$
for all $(s,\alpha)\in S\times A$.
We say that action $\alpha$ is \emph{enabled} in state $s$ iff $\sum_{t\in S}\mdptransitions(s,\alpha,t) =1$ and denote the set of all actions that are enabled in state $s$ by $A(s)$. We further require that $A(s) \not=\emptyset$ for all $s\in S$.
If for a state $s$ and all actions $\alpha\in A(s)$, we have $\mdptransitions(s,\alpha,s)=1$, we say that $s$ is \emph{absorbing}.

The paths of $\MDP$ are finite or infinite sequences $s_0 \, \alpha_0 \, s_1 \, \alpha_1  \ldots$ where states and actions alternate such that $\mdptransitions(s_i,\alpha_i,s_{i+1}) >0$ for all $i\geq0$.
For $\finitepath = s_0 \, \alpha_0 \, s_1 \, \alpha_1 \,  \ldots \alpha_{k-1} \, s_k$, the probability of $\finitepath$ is
$\mdptransitions(\finitepath) = \mdptransitions(s_0,\alpha_0,s_1) \cdot \ldots \cdot \mdptransitions(s_{k-1},\alpha_{k-1},s_k)$ and $\last{\finitepath}=s_k$ is the last state of $\finitepath$.
We write $\finitepath_i = s_i$ for the $i$-th state of $\finitepath$.
For two (finite) paths $\finitepath$ and $\infinitepath$ with $\finitepath$ ending in the state that $\infinitepath$ begins with, we write $\finitepath \circ \infinitepath$ for their concatenation.

We will equip MDPs with a reward function $\rew\colon S \to \Reals$.
Moreover, in the following we often consider the logarithmic rewards ${\log} \circ \reward$.
To slightly ease notation, we write $\log(\reward)$.

For algorithmic problems, we require all transition probabilities and rewards to be rational.
The \emph{size} of $\MDP$ is the sum of the number of states plus the total sum of the encoding lengths in binary of the non-zero probability values $\mdptransitions(s,\alpha,s')$ as fractions of (co-prime) integers as well as the encoding length in binary of the rewards if a reward function is used.

A \emph{Markov chain} is an MDP in which the set of actions is a singleton. In this case, we can drop the set of actions and consider a Markov chain as a tuple $\MDP=(S,\mctransitions)$ where $\mctransitions$ now is a function from $S\times S$ to $[0,1]$.

An \emph{end component (EC)} of $\MDP$ is a strongly connected sub-MDP formalized by a subset $S^\prime\subseteq S$ of states and a non-empty subset $\mathfrak{A}(s)\subseteq A(s)$  for each state $s\in S^\prime$ such that for each $s\in S^\prime$, $t\in S$ and $\alpha\in \mathfrak{A}(s)$ with $\mdptransitions(s,\alpha,t)>0$, we have $t\in S^\prime$ and such that in the resulting sub-MDP all states are reachable from each other. Maximal end components (MECs) are end components that are not contained in any other end component. The set of all MECs is computable in polynomial time (see, e.g., \cite{DBLP:books/daglib/0020348}).
A \emph{bottom strongly connected component (BSCC)} is an end component with exactly one action per state, i.e.\ an end component that is a Markov chain.


\paragraph*{Schedulers}
A \emph{ scheduler} (also called \emph{policy}, \emph{strategy}) for $\MDP$ is a function $\sched$ that assigns to each finite path $\finitepath$ a probability distribution over $\Actions(\last{\finitepath})$.
If  $\sched(\finitepath)=\sched(\finitepath^\prime)$ for all finite paths $\finitepath$ and $\finitepath^\prime$ with $\last{\finitepath}=\last{\finitepath^\prime}$, we say that $\sched$ is \emph{memoryless}.
In this case, we also view  schedulers as functions mapping states $s\in S$ to probability distributions over $\Actions(s)$.
A scheduler $\sched$ is called \emph{deterministic} if $\sched(\finitepath)$ is a Dirac distribution for each finite path $\finitepath$, in which case we also view the scheduler as a mapping to actions in $A(\last{\finitepath})$.
We write MD-scheduler for memoryless deterministic scheduler.

\paragraph*{Probability measure}
We write $\ProbabilityMDP<\MDP, s><\sched>$ to denote the probability measure induced by a scheduler $\sched$ and a state $s$ of an MDP $\MDP$.
It is defined on the $\sigma$-algebra generated by the cylinder sets $\Cyl(\finitepath)$ of all infinite extensions of a finite path  $\finitepath = s_0 \, \alpha_0 \, s_1 \, \alpha_1 \,  \ldots \alpha_{k-1} \, s_k$, by assigning to $\Cyl(\finitepath)$ the probability that $\finitepath$ is realized under $\sched$.
This is $\ProbabilityMDP<\MDP, s><\sched>[\Cyl(\finitepath)] = \prod_{i=0}^{k-1} \sched(s_0\, \alpha_0 \dots \, s_i )(\alpha_i) \cdot \mdptransitions(s_i,\alpha_i,s_{i+1})$ if $s_0=s$ and $0$ otherwise.
This can be extended to a unique probability measure on the mentioned $\sigma$-algebra.
For details, see \cite{DBLP:books/daglib/0020348}.
For a random variable $X$, i.e.\ a measurable function defined on infinite paths in $\MDP$, we denote the expected value of $X$ under a scheduler $\sched$ and state $s$ by $\ExpectationMDP<\MDP, s><\sched>[X]$.
We define
$\ExpectationMDP<\MDP, s><\min>[X] = \inf_{\sched} \ExpectationMDP<\MDP, s><\sched>[X]$
and
$\ExpectationMDP<\MDP, s><\max>[X] = \sup_{\sched} \ExpectationMDP<\MDP, s><\sched>[X]$.

\subsection*{Objectives}
We recall standard objectives for Markov models.
All appearing sets / functions are measurable, see e.g.\ \cite{DBLP:books/wi/Puterman94}.

\paragraph*{Reachability}
For a set of states $T$, $\lozenge T = \{\infinitepath \mid \exists i. \infinitepath_i \in T\}$ contains all paths that eventually reach a state in $T$.
For a single state $t$, we also write $\lozenge t$ instead of $\lozenge \{t\}$.

\paragraph*{Total Reward}
For a reward function $\reward \colon \States \to \Reals$, the total reward of a path $\infinitepath$ is defined as $\totalreward(\infinitepath) = \liminf_{n \to \infty} \sum_{i = 0}^n \reward(\infinitepath_i)$.
(In general, the limit might not be defined and we arbitrarily choose the infimum.)

\paragraph*{Mean Payoff}
For a reward function $\reward \colon \States \to \Reals$, we define the two versions of the mean payoff of a path $\infinitepath$ as
\begin{align*}
	\MPsup_{\reward}(\infinitepath) & = {\limsup}_{n\to \infty} \tfrac{1}{n+1} {\sum}_{i=0}^n \reward(\infinitepath_i) \quad \text{and} \\
	\MPinf_{\reward}(\infinitepath) & = {\liminf}_{n\to \infty} \tfrac{1}{n+1} {\sum}_{i=0}^n \reward(\infinitepath_i).
\end{align*}
In finite Markov chains, almost all paths $\infinitepath$ satisfy $\MPsup_{\reward}(\infinitepath) = \MPinf_{\reward}(\infinitepath)$.
Hence, we simply write $\MP_{\reward}$ instead of $\MPsup_{\reward}$ or $\MPinf_{\reward}$ when the difference does not matter.
In a strongly connected Markov chain $\MC$, we additionally have $\ExpectationMC<\MC, \initialstate>[\MP_{\reward}] = \ExpectationMC<\MC, \initialstate'>[\MP_{\reward}]$ for any pair of states $\initialstate, \initialstate' \in \States$ (e.g., \cite[Chp.~9]{DBLP:books/wi/Puterman94}).

\subsection*{Comparison of Succinctly Represented Integers (CSRI)}

The \emph{comparison of succinctly represented integers (CSRI)} problem will play an important role:
Given natural numbers $a_1,\dots,a_n, b_1,\dots,b_n,c_1,\dots,c_m,d_1,\dots,d_m\in \mathbb{N}$ (in binary), decide whether
$
a_1^{b_1}\cdot \dots \cdot a_n^{b_n} > c_1^{d_1}\cdot \dots \cdot c_m^{d_m}
$.
This problem is not known to be solvable in polynomial time  \cite{DBLP:journals/toct/EtessamiSY14}.
However, \cite{DBLP:journals/toct/EtessamiSY14} shows that the problem is in P if either
	(i)~Baker's refinement \cite{Baker+1998+37+44} of the \emph{ABC conjecture} of Masser and Oesterl\`e, or
	(ii)~the Lang-Waldschmidt conjecture \cite{lang1978elliptic} hold.
	
Unconditionally, an upper bound for CSRI is $\exists \Reals$ (see \Cref{prop:CSRI} in \Cref{app:CSRI}), the complexity class of the existential theory of the reals, which is contained in PSPACE~\cite{DBLP:conf/stoc/Canny88}.
Interestingly, deciding \emph{equality} of succinctly represented integers (ESRI), i.e.\ asking whether $a_1^{b_1}\cdot \dots \cdot a_n^{b_n} = c_1^{d_1}\cdot \dots \cdot c_m^{d_m}$ for an instance as above, is unconditionally in P \cite{DBLP:journals/toct/EtessamiSY14}.

Several of our considered problems involve such comparisons.
We will hence use CSRI to also denote the class of all problems that can be reduced to CSRI in polynomial time.
Further, we will write $\text{P}^\text{CSRI}$ to denote the class of problems decidable in polynomial time with an oracle for CSRI and analogously $\text{NP}^{\text{CSRI}}$ for the non-deterministic polynomial time variant.
Note that if one of the conjectures (i) or (ii) mentioned above holds,  $\text{CSRI}=\text{P}^\text{CSRI} = \text{P}$.
Unconditionally, we have $\text{CSRI}\subseteq \exists \mathbb{R}$ and $\text{P}^\text{CSRI} \subseteq \text{NP}^\text{CSRI} \subseteq \text{PSPACE}$.

\section{Markov Chains}
\label{sec:MC}

Let $\MC = (\States, \mctransitions)$ be a Markov chain and $\reward \colon \States \to \Reals_{>0}$ a reward function.
We exclude reward zero here for simplicity; states with reward zero can be collapsed to one absorbing state with reward $0$ as reaching such a state immediately leads to a multiplicative reward of $0$.
The presence of such an absorbing state with reward $0$ does not influence any of the results as also explained later in \Cref{rem:zero}.

For a \emph{finite} path $\finitepath = s_0 s_1 \dots s_{n+1}$, we write $\multreward_{\reward}(\finitepath) = \prod_{k=0}^n \reward(\finitepath_k)$ for its \emph{multiplicative reward} (note that rewards are obtained when leaving a state).
For an infinite path $\infinitepath$, we define the (lower and upper) multiplicative reward as
\begin{align*}
	\pathreward(\infinitepath) & \coloneqq {\limsup}_{n\to\infty} {\prod}_{k \leq n} \reward(\infinitepath_k) \quad \text{ and } \\
	\pathrewardinf(\infinitepath) & \coloneqq {\liminf}_{n\to\infty} {\prod}_{k \leq n} \reward(\infinitepath_k).
\end{align*}
Note that these values need not equal each other.
More so, we later see that they are different in many cases.
We are interested in the expected value of these random variables when starting from a given state $\initialstate \in \States$, i.e.\ we want to determine $\ExpectationMC<\MC, \initialstate>[\pathreward]$ and $\ExpectationMC<\MC, \initialstate>[\pathrewardinf]$, which we call the \emph{values (of the state $\initialstate$)}.
The results of this section can be summarized as follows:
\begin{theorem}[Main result for MC]
\label{thm:mainresultMC}
Let $\MC=(\States,\mctransitions)$ a Markov chain, $\initialstate \in \States$ an initial state, and $\reward\colon \States\to \Rationals_{>0}$ a (rational) reward function.
Then:
\begin{enumerateindent}
	\item
	The values $\ExpectationMC<\MC, \initialstate>[\pathreward] $ and $\ExpectationMC<\MC, \initialstate>[\pathrewardinf]$ can be computed in polynomial time with an oracle for CSRI, i.e.\ the computation is possible in polynomial space and, if CSRI is in P, the computation is possible in polynomial time.
		\item
	If an absorbing state is reached almost surely in  $\MC$, then $\ExpectationMC<\MC, \initialstate>[\pathreward] = \ExpectationMC<\MC, \initialstate>[\pathrewardinf]$ can be computed in polynomial time (independent of CSRI).
\end{enumerateindent}
\end{theorem}

\begin{remark}
Interestingly, the part for which we cannot show an unconditional polynomial time upper bound is deciding whether the value of a state in a BSCC is $0$ or $\infty$ when we already know that it is one of the two.
In fact, we show that this problem is as hard as CSRI (\Cref{prop:hardnessMP}).
Whether the value of a state in a BSCC is either strictly between $0$ and $\infty$ or is one of the two can indeed be decided in polynomial time (relating to ESRI).
\end{remark}

To prove these results, we first focus on strongly connected Markov chains.
For these, we provide necessary and sufficient conditions for the expected multiplicative reward to be $0$ or $\infty$ and show how these conditions can be checked.
Then, we can assign values to the states in BSCCs of a general Markov chain.
For the transient states, we provide a linear system of equations such that the values of the states $s\in \States$ form the unique solution to the system if these values are finite and such that the system has no non-negative solution otherwise.

All proofs for this section can be found in \Cref{appendix:proofs}.

\subsection{Strongly Connected Markov Chains} \label{sec:mc:scc}

The first main question in strongly connected Markov chains is whether the expected multiplicative reward is $0$, $\infty$, or a finite (positive) value.

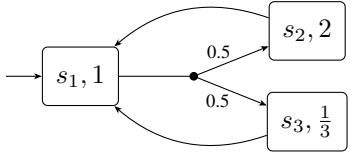
\begin{figure}
	\centering
	\begin{tikzpicture}[auto,initial text=]
		\node[state,initial=left] at (1, -.6) (s1) {$s_1, 1$\strut};
		\node[actionnode] at (2.5, -.6) (a11) {};
		\node[state] at (4, 0) (s2) {$s_2, 2$\strut};
		\node[state] at (4, -1.2) (s3) {$s_3, \tfrac13$\strut};


		\path[->,directedge]
			(s2) edge[bend right] (s1)
			(s3) edge[bend left] (s1)
		;
		\path[->,probedge]
			(a11) edge node[prob] {0.5} (s2)
			(a11) edge node[swap,prob] {0.5} (s3)
		;
		\path[-,actionedge]
			(s1) edge  (a11)
		;
	\end{tikzpicture}
	\caption{
		The Markov chain  $\MC$ used in Remark \ref{rem:infinitevszero} with state names and rewards depicted in the states.
	}
	\label{fig:infinitevszero}
\end{figure}

\begin{remark}
\label{rem:infinitevszero}
We point out a possibly counter-intuitive phenomenon in this context.
Let $\multreward_{\reward}^n(\infinitepath) = \prod_{k \leq n} \reward(\infinitepath_k)$ the multiplicative reward of $\infinitepath$ accumulated in the first $n$ steps.
It is possible that $\ExpectationMC<\MC, \initialstate>[\multreward_{\reward}^n] \to \infty$ for $n \to \infty$ while at the same time $\ExpectationMC<\MC, \initialstate>[\pathreward]=0$.
Consider the Markov chain $\MC$ in \Cref{fig:infinitevszero}.
The expected multiplicative reward along the simple cycles from $s_1$ to $s_1$ is $\frac12 \cdot 2 +\frac12\cdot \frac13 =\frac76 > 1$.
As the multiplicative reward received during one loop from $s_1$ to $s_1$ is independent from all other loops, the expected multiplicative reward after $n$ loops, i.e.\ after $2n$ steps, is $\ExpectationMC<\MC, s_1>[\multreward_{\reward}^{2n}] = \ExpectationMC<\MC, s_1>[\multreward_{\reward}^2]^n = \left(\frac76\right)^n \to \infty$.
In contrast, almost all infinite paths have a multiplicative reward $\pathreward$ of $0$.
The intuitive reason is that along a path the factors $2$ and $\frac13$ occur equally frequently in the long-run almost surely.
Observe that when considering the expectation of $\multreward$, we first take the limit and then compute the expectation, whereas the above considers the limit of the expectation.
Formally, this claim follows from \cref{lem:meanpayoffMC} below.
\end{remark}

\subsubsection{Characterizing Values}
Let $\MC=(\States,\mctransitions)$ be a strongly connected Markov chain with a reward function $\reward\colon\States\to \Reals_{>0}$.
We start by establishing necessary and sufficient conditions for the expected value of $\pathreward$ and $\pathrewardinf$ to be $0$ or $\infty$; computational issues will be discussed later.
We consider the \emph{additive} mean payoff with respect to the reward function $\log(\reward)$, which allows us to adapt some of the existing theory into our setting.
We later discuss the algorithmic approach to deal with these logarithms.
For additive rewards, it is well-known that almost all paths in a strongly connected Markov chain achieve the same \emph{mean payoff}.
\begin{restatable}{lemma}{lemmeanpayoffMC} \label{lem:meanpayoffMC}
	Let $\MC$ and $\reward$ be as above and $\initialstate \in \States$ any state.
	Then,
	1)~if $\ExpectationMC<\MC, \initialstate>[\MP_{\log(\reward)}] < 0$, then $\ExpectationMC<\MC, \initialstate>[\pathreward] = 0$ and hence $\ExpectationMC<\MC, \initialstate>[\pathrewardinf] = 0$, and
	2)~if $\ExpectationMC<\MC, \initialstate>[\MP_{\log(\reward)}] > 0$, then $\ExpectationMC<\MC, \initialstate>[\pathrewardinf] = \infty$ and hence $\ExpectationMC<\MC, \initialstate>[\pathreward] = \infty$.
\end{restatable}
The remaining case is that the expected logarithmic mean payoff is zero, i.e.\ $\ExpectationMC<\MC, \initialstate>[\MP_{\log(\reward)}] = 0$.
Here, we actually have to distinguish two cases: either the multiplicative reward of every cycle in $\MC$ is $1$ or there are cycles with a multiplicative reward different to $1$, as we illustrate in the following.
\begin{figure}[t]
	\centering
	\begin{tikzpicture}[auto,initial text=]
		\node[state,initial=left] at (0,0) (s1) {$s_1, 2$\strut};
		\node[state] at (1.5,0) (s2) {$s_2, \tfrac{1}{2}$};

		\path[->,directedge]
			(s1) edge[bend left] (s2)
			(s2) edge[bend left] (s1)
		;
		\node at (0,.75) {};
		\node at (0,-.75) {};
	\end{tikzpicture}\quad%
	\begin{tikzpicture}[auto,initial text=]
		\node[state,initial=left] at (0,0) (s1) {$s_1, 2$\strut};
		\node[state] at (2,0) (s2) {$s_2, \tfrac{1}{2}$};

		\node[actionnode] at (1, 0.5) (a11) {};
		\node[actionnode] at (1, -0.5) (a21) {};

		\path[->,probedge]
			(a11) edge[out=80,in=120,pos=0.8] node[prob] {0.5} (s2)
			(a11) edge[out=100,in=60,pos=0.8,swap] node[prob] {0.5} (s1)
			(a21) edge[out=-100,in=-60,pos=0.8] node[prob] {0.5} (s1)
			(a21) edge[out=-80,in=-120,pos=0.8,swap] node[prob] {0.5} (s2)
		;
		\path[-,actionedge]
			(s1) edge[out=10,in=-100] (a11)
			(s2) edge[out=190,in=80] (a21)
		;
		\node at (0,.75) {};
		\node at (0,-.75) {};
	\end{tikzpicture}
	\caption{
		Two Markov chains that both achieve a logarithmic mean payoff of zero, but differ in their achieved multiplicative reward.
	} \label{fig:logzeromp_example}
\end{figure}
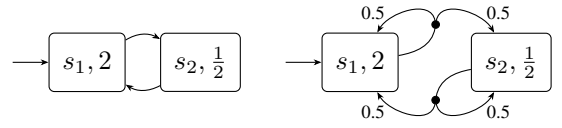
\begin{example}
	Consider the two Markov chains in \cref{fig:logzeromp_example}.
	In both, the expected logarithmic mean payoff is zero.
	However, the left MC obtains a finite multiplicative reward, namely $2$ for $\pathreward$ and $1$ for $\pathrewardinf$ as the sequence of multiplicative rewards of prefixes of the unique path is $2,1,2,1,\dots$.
	In contrast, in the right MC, we have that $\ExpectationMC<\MC, \initialstate>[\pathreward] = \infty$ and $\ExpectationMC<\MC, \initialstate>[\pathrewardinf] = 0$:
	Intuitively, the sum of logarithmic rewards over prefixes of an individual path in this MC is a random walk without drift, which eventually obtains arbitrarily large and arbitrarily small values almost surely.
	Consequently, the supremum and infimum of multiplicative rewards tend to $\infty$ and $0$ respectively.
	(Formal treatment is provided below.)
\end{example}
In the treatment of this remaining case, cycles with multiplicative reward equal to or different from 1 will play a central role, inviting an abbreviation.
\begin{definition}
\label{def:n-cycle}
	A cycle $s_0 s_1 \dots s_{k+1}$ ($s_0 = s_{k+1}$) is called a \emph{1-cycle} if $\prod_{i=0}^k \reward(s_i) = 1$ and \emph{n-cycle} otherwise.
\end{definition}
Note that a cycle is a 1-cycle if and only if the sum of the logarithms of the rewards on the cycle is $0$.
With this definition at hand, we can now capture the effect of n-cycles formally.
\begin{restatable}{lemma}{zeromeanpayoff}
\label{lem:zeromeanpayoff}
	Let $\MC$ and $\reward$ be as above and $\initialstate \in \States$ any state.
	Assume that $\ExpectationMC<\MC, \initialstate>[\MP_{\log(\reward)}] = 0$.
	Then:
		1)~If all cycles 
		are 1-cycles, then for all states $\initialstate \in \States$ we have $0 < \ExpectationMC<\MC, \initialstate>[\pathrewardinf] \leq \ExpectationMC<\MC, \initialstate>[\pathreward] < \infty$.
		2)~If there is an n-cycle, 
		then for all states $\initialstate \in S$ we have $\ExpectationMC<\MC, \initialstate>[\pathrewardinf] = 0$ and $\ExpectationMC<\MC, \initialstate>[\pathreward] = \infty$.
\end{restatable}

Thus, the expected multiplicative reward in $\MC$ is only finite and positive if all cycles have multiplicative reward $1$.
In this case, the expected multiplicative reward can be directly characterized as shown in the following two lemmas.
\begin{restatable}{lemma}{uniquereward} \label{stm:unique_reward}
	Let $\MC$ and $\reward$ be as above.
	All cycles in $\MC$ are 1-cycles if and only if for all pairs of states $s,t$ and all paths $\finitepath$ and $ \finitepath' $ from $s$ to $t$, we have $\multreward(\finitepath) = \multreward(\finitepath')$.
\end{restatable}
Thus, when there are no n-cycles, then for every pair of states $s$ and $t$ there is a unique multiplicative reward for all paths from $s$ to $t$, which we denote by $R(s,t)$.

\begin{restatable}{lemma}{finitevaluesMC}
\label{lem:finitevaluesMC}
	Let $\MC$ and $\reward$ be as above, and assume that all cycles in $\MC$ have multiplicative reward $1$.
	Let $\initialstate$ as well as $R(s,t)$ for all $s,t\in S$ be as above.
	Then, $0 < \ExpectationMC<\MC, \initialstate>[\pathrewardinf] = \min_{t\in \States} R(\initialstate,t) \leq 1 \leq  \max_{t\in \States} R(\initialstate,t) = \ExpectationMC<\MC, \initialstate>[\pathreward] < \infty$.
\end{restatable}

Note that contrary to many classical objectives, this also means that not all states in a BSCC have the same value.

\subsubsection{Computing Multiplicative Rewards}
For the computation of expected multiplicative rewards, we restrict to inputs with rational transition probabilities and rewards, given as fractions of (co-prime) integers in binary.
So, let $\MC=(\States,\mctransitions)$ be a strongly connected Markov chain with rational transition probabilities and with a reward function $\reward\colon\States\to \Rationals_{>0}$ and let
$\initialstate\in \States$ be an initial state.
From the above, the following steps are necessary to compute $ \ExpectationMC<\MC, \initialstate>[\pathreward]$:
\begin{enumerateindent}
\item
Decide whether $\ExpectationMC<\MC, \initialstate>[\MP_{\log(\reward)}]$ is $< 0$, $= 0$, or $> 0$.
In the first case and third case, both values of $\initialstate$ are $0$ or $\infty$, respectively, 
by \Cref{lem:meanpayoffMC}.
\item
If $\ExpectationMC<\MC, \initialstate>[\MP_{\log(\reward)}] = 0$, check whether there are n-cycles.
If yes, $\ExpectationMC<\MC, \initialstate>[\pathreward]=\infty$ and $\ExpectationMC<\MC, \initialstate>[\pathrewardinf]=0$ by \Cref{lem:zeromeanpayoff}.
\item
Otherwise, all cycles have multiplicative reward $1$, and $\ExpectationMC<\MC, \initialstate>[\pathreward] = \max_{t\in \States} R(\initialstate,t)$ and $\ExpectationMC<\MC, \initialstate>[\pathrewardinf] = \min_{t\in \States} R(\initialstate,t)$ with $R(\initialstate,t)$ as in \Cref{lem:finitevaluesMC}.
\end{enumerateindent}

\paragraph*{Step~1.\ Comparing $\ExpectationMC<\MC, \initialstate>[\MP_{\log(\reward)}] $ to $0$}
In strongly connected Markov chains, the mean payoff can be expressed in terms of the unique stationary distribution (see, e.g., \cite[Sec.~2.6.2]{Kulkarni2011}).
This stationary distribution $\theta\in [0,1]^S$ can be computed by solving a linear equation system using the transition probabilities of $\MC$ in polynomial time (see, e.g., \cite{DBLP:books/daglib/0020348}).
So, we can obtain  fractions $\theta_s=a_s/b_s$ of integers in binary for all states $s\in S$ in polynomial time (and hence of polynomially sized binary representation).
We have 
$
\ExpectationMC<\MC, \initialstate>[\MP_{\log(\reward)}]  = {\sum}_{s\in \States} \theta_s \cdot \log(\reward(s))
$.
Now, applying the exponential function to this sum, we see that $\ExpectationMC<\MC, \initialstate>[\MP_{\log(\reward)}] \bowtie 0$ for ${\bowtie} \in \{<,=,>\}$ if and only if $\prod_{s\in S} \reward(s)^{\theta_s} \bowtie 1$.
This inequality holds if and only if $\prod_{s\in S} \reward(s)^{a_s} \bowtie \prod_{s\in S} \reward(s)^{b_s}$.
Letting $d$ be the least common multiple of the denominators of the values $\reward(s)$ with $s\in S$, we can multiply this inequality with $d^{(\max_s a_s\cdot b_s)}$  to obtain an instance of CSRI for ${\bowtie} \in \{<,>\}$ and an instance of ESRI if ${\bowtie}$ is $=$.
We conclude:
\begin{proposition}
\label{prop:logMP_MC}
	Deciding whether $\ExpectationMC<\MC, \initialstate>[\MP_{\log(\reward)}] =0$ is in P.
	Deciding whether $\ExpectationMC<\MC, \initialstate>[\MP_{\log(\reward)}] \bowtie 0$ for $\bowtie{} \in \{<,>\}$ is in the complexity class ${\text{CSRI}}\subseteq \exists \mathbb{R}$.
\end{proposition}
So, as described in \Cref{sec:prelim}, the latter problem is in P if either of the two mentioned number-theoretic conjectures hold. 
Further, distinguishing divergence to $\infty$ from convergence to $0$ is at least as hard as CSRI:
\begin{restatable}{proposition}{prophardnessMP}
\label{prop:hardnessMP}
Deciding whether $\ExpectationMC<\MC, \initialstate>[\MP_{\log(\reward)}] > 0$ is at least as hard as CSRI.
\end{restatable}

\paragraph*{Step~2.\ Checking whether all cycles have multiplicative reward $1$}
By \Cref{stm:unique_reward}, all cycles have multiplicative reward of $1$ if and only if for any two states $s$ and $t$ all paths from $s$ to $t$ have the same reward.
Since the Markov chain is strongly connected, this is equivalent to all paths from $\initialstate$ to $t$ having the same reward.
This in turn can be determined by standard graph search algorithms in polynomial time.
For example, we can follow a simple breadth-first search:
Starting from $\initialstate$, we store for each state $t$ the reward of the path through which we reached $t$ from $\initialstate$ (when moving from $t$ to $t'$, the value for $t'$ is simply the value of $t$ times $\reward(t')$).
If we encounter a successor of $t$ which we already visited, however now find a path with different weight, we return \enquote{false}.
When the search completes without error, we return \enquote{true}.
As we only follow each edge once, we only need to perform linearly many multiplications, and all involved numbers are of polynomial size.

\paragraph*{Step~3.\ Computing the values $R(s,t)$ for states $s$ and $t$}
These are computed as a by-product of the above algorithm.

\subsection{General Markov Chains}

Now, we address the problem of finding the expected multiplicative reward in general, i.e.\ not necessarily strongly connected, Markov chains.
We consider the multiplicative reward $\pathreward$ defined using the $\limsup$, $\pathrewardinf$ is analogous.

Let $\MC = (\States, \mctransitions)$ be a Markov chain and $\initialstate \in \States$ an initial state.
We assume w.l.o.g.\ that all states are reachable from $\initialstate$.
For a state $s \in S$, we denote the \emph{value of $s$} by $v_s = \ExpectationMC<\MC, s>[\pathreward]$.
First, the results from the previous section provide a characterization of the values $v_s$ for all states in the BSCCs of $\MC$, which are strongly connected Markov chains.
Thus, let $T$ be the set of states that are in BSCCs, and $v_t$ the values for all $t \in T$.
If any of these values are $\infty$, we also have $v_{\initialstate} = \infty$:
We assumed that all states are reachable, i.e.\ in particular $t$, and all rewards are positive, hence there is a path with non-zero reward from $\initialstate$ to $t$.
So, next assume that the value $v_t$ is finite for all states $t \in T$.
Further, let $T_0 \coloneqq \{t \in T \mid v_t = 0\}$ be the set of states in a BSCC with value $0$.
(Observe that either all states or no state of a BSCC have value $0$.)
\begin{restatable}{lemma}{lemzeros} \label{lem:zeros}
	For each $s \in \States$, we have $v_s = 0$ if and only if $\ProbabilityMC<\MC, s>[\reach T_0] = 1$.
\end{restatable}

\cref{lem:zeros} allows us to compute the set $\States_0 \coloneqq \{s \in \States \mid v_s = 0\}$ given the values $v_t$ for all $t\in T$ via standard graph analysis.
The remaining states for which we aim to determine the values are the states in the set $\States_? \coloneqq S \setminus (\States_0 \union T)$.
For these, we consider the following system of equations.
For variables $x_s$ for all states $s \in \States_?$, we require that
\begin{equation}
	x_s = \reward(s) \cdot \left( {\sum}_{s' \in S_?} \mctransitions(s,s') \cdot x_{s'} + {\sum}_{t \in T} \mctransitions(s,t) \cdot v_t \right). \label{eq:transient}
\end{equation}
We show that a non-negative solution exists if and only if $v_{\initialstate} = \ExpectationMC<\MC, \initialstate>[\pathreward] < \infty$, and the expected multiplicative reward is the \emph{unique} solution in this case.
To simplify matters, we first restrict to non-negative solutions, i.e.\ $x_s \geq 0$.
\begin{restatable}{lemma}{lemmcnonnegsolution}
	If $\ExpectationMC<\MC, \initialstate>[\pathreward] < \infty$, then $v_s$ 
	is a (non-negative) solution to \cref{eq:transient}.
\end{restatable}
\begin{corollary}
	If there is no non-negative solution to \cref{eq:transient}, then $\ExpectationMC<\MC, \initialstate>[\pathreward] = \infty$.
\end{corollary}
Next, we prove that if there is a non-negative solution to \cref{eq:transient}, then the expected multiplicative reward is finite.
In fact, the vector of values for each state is the component-wise least solution to these constraints.
As the values for states in $\States_?$ are strictly positive, this shows that $(v_s)_{s \in \States_?}$ is the unique solution.
To prove that the vector of values for each state is the component-wise least solution, we define the update operator $U \colon \RealsNonneg^{\States_?} \to \RealsNonneg^{\States_?} $ derived from the constraint as $U \colon (x_s)_{s \in \States_?} \mapsto (y_s)_{s \in \States_?}$ where $y_s = \reward(s) \cdot \left({\sum}_{s' \in \States_?} \mctransitions(s, s') \cdot x_{s'} + {\sum}_{t \in T} \mctransitions(s,t) \cdot v_t \right)$.
Any non-negative solution to \cref{eq:transient} is a fixed point of $U$ and vice-versa.

\begin{restatable}{lemma}{lemmcvalueiteration} \label{thm:valueiteration}
	We have $(v_s)_{s \in \States_?} = \lim_{n\to \infty} U^n(0)$.
\end{restatable}
\begin{restatable}{corollary}{cormcleast} \label{cor:least}
	Suppose there is a non-negative solution $x^\ast$ to \cref{eq:transient}. 
	Then, for all $s \in S$, we have $v_s \leq x^{\ast}_s$.
\end{restatable}
\begin{restatable}{corollary}{cormcunique}
	If $\ExpectationMC<\MC, \initialstate>[\pathreward] < \infty$, then $(v_s)_{s \in \States_?}$ is the unique solution to \cref{eq:transient}.
\end{restatable}
In summary, we obtain:
\begin{theorem}
	Given the values $v_t$ for all $t\in T$,
	either \cref{eq:transient} has no non-negative solution and $\ExpectationMC<\MC, \initialstate>[\pathreward] = \infty$,
	or \cref{eq:transient} has a unique non-negative solution, which is a strictly positive vector of finite values, and this solution is $(v_s)_{s\in S_?}$.
\end{theorem}
So, after analyzing BSCCs as described in the previous section in polynomial time with a CSRI-oracle, all remaining steps can be carried out in polynomial time.
If an absorbing state is reached almost surely, there are no BSCCs except for the absorbing states and so the values can be computed in polynomial time. Furthermore, if an absorbing state is reached almost surely, it is reached almost surely in finitely many steps so that $\pathreward$ and $\pathrewardinf$ almost surely coincide. This
 finishes the proof of \Cref{thm:mainresultMC}.

We conclude with the following observation: When \cref{eq:transient} has no non-negative solution, it can happen that the equation has no, one, or infinitely many solutions.
\begin{example}
	\emph{No solution:} Consider an MC with two states $x$ and $t$, transitions $\mctransitions(x, x) = \mctransitions(x, t) = \frac{1}{2}$ and $\mctransitions(t, t) = 1$, and rewards $\reward(x) = 2$, $\reward(t) = 1$.
	State $t$ comprises the only BSCC and is assigned a value of $v_t = 1$.
	\cref{eq:transient} for $s$ yields $x = x + 1$, which has no (finite) solution.

	\emph{One solution:} Modify the above MC to $\reward(x) = 4$, yielding the equation $x = 2 x + 1$, which has $x = -1$ as unique solution.

	\emph{Infinitely many solutions:} Consider an MC with three states $x$, $y$, and $t$, transitions $\mctransitions(x, x) = \mctransitions(y, y) = \frac{1}{2}$, $\mctransitions(x, y) = \mctransitions(x, t) = \mctransitions(y, x) = \mctransitions(y, t) = \frac{1}{4}$, and $\mctransitions(t, t) = 1$, and rewards $\reward(x) = \reward(y) = 4$ and $\reward(t) = 1$.
	As before, $t$ is the only BSCC and has value $v_t = 1$.
	\cref{eq:transient} yields $x = 2 x + y + 1$ and $y = 2 y + x + 1$, and any $x + y = -1$ is a solution.
\end{example}
Finally, we remark how to treat states with reward $0$.
\begin{remark}
\label{rem:zero}
	W.l.o.g., states with $\reward(s)=0$ can be made absorbing:
	As soon as such a state is visited, any path can only obtain a reward of zero.
	After removing any now unreachable states, we can proceed as above, treating zero-states just like BSCCs with reward 0 (i.e.\ $v_s = 0$) when building \cref{eq:transient}.
\end{remark}

\section{Markov Decision Processes}

In this section, we turn our attention to MDPs.
Let $\MDP = (\States, A, \mdptransitions)$ be an MDP, $\reward \colon \States \to \Reals_{>0}$ a reward function, and $\initialstate \in \States$ an initial state.
The multiplicative rewards of paths are defined as for Markov chains.
We investigate the maximization problem for the expected multiplicative reward. Given an MDP $\MDP$, reward function $\reward$, and initial state $\initialstate$, we are interested in the following quantities
where $\Strategies<\MDP>$  is the set of schedulers:
\begin{align*}
	\ExpectationMDP<\MDP, \initialstate><\max>[\pathreward] & \coloneqq {\sup}_{\strategy \in \Strategies<\MDP>} \ExpectationMDP<\MDP, \initialstate><\strategy>[\pathreward] \quad \text{ and } \\
	\ExpectationMDP<\MDP, \initialstate><\max>[\pathrewardinf] & \coloneqq {\sup}_{\strategy \in \Strategies<\MDP>} \ExpectationMDP<\MDP, \initialstate><\strategy>[\pathrewardinf]
\end{align*}
For minimization, we can consider the reward function $\reward'$ where $\reward'(s) = \frac{1}{\reward(s)}$, switch between $\pathreward$ and $\pathrewardinf$, and then invert the result again (observe that $\reward(s) > 0$ by assumption).
This is analogous to negating rewards in the additive case.

The results of this section can be summarized as follows:
\begin{theorem}[Main result for MDP]
Let $\MDP=(\States, A, \mdptransitions)$ be  an MDP, $\initialstate\in \States$ an initial state, and $\reward\colon \States \to \Rationals_{>0}$ a (rational) reward function. Then:
\begin{enumerateindent}
\item The values $\ExpectationMDP<\MDP, \initialstate><\max>[\pathreward]$  and  $\ExpectationMDP<\MDP, \initialstate><\max>[\pathrewardinf] $ can be computed in polynomial space.
\item If in $\MDP$ an absorbing state is reached almost surely under any scheduler, the values $\ExpectationMDP<\MDP, \initialstate><\max>[\pathreward] $ and $\ExpectationMDP<\MDP, \initialstate><\max>[\pathrewardinf] $
are equal and can be computed in polynomial time.
\item There is an MD-scheduler $\sched$ such that $\ExpectationMDP<\MDP, \initialstate><\max>[\pathreward] = \ExpectationMDP<\MDP, \initialstate><\sched>[\pathreward]$.
\item If  $\ExpectationMDP<\MDP, \initialstate><\max>[\pathrewardinf] < \infty$, there is an MD-scheduler $\sched$ such that $\ExpectationMDP<\MDP, \initialstate><\max>[\pathrewardinf] = \ExpectationMDP<\MDP, \initialstate><\sched>[\pathrewardinf]$.
\item
The threshold problem -- deciding whether $\ExpectationMDP<\MDP, \initialstate><\max>[\pathreward]\geq \vartheta$ (and $\ExpectationMDP<\MDP, \initialstate><\max>[\pathrewardinf]\geq \vartheta$) for a given threshold $\vartheta \in \Rationals$ -- is in $\text{NP}^{\text{CSRI}}$ and hence in NP if $\text{CSRI}\subseteq \text{P}$.
\item Additionally given an absorbing target state $t\in \States$, the \emph{multiplicative stochastic shortest path problem} asking for the value
$\sup_{\sched}\ExpectationMDP<\MDP, \initialstate><\sched>[\pathreward]$ where the supremum ranges over all schedulers $\sched$ that reach $t$ almost surely can be solved in polynomial space.
\end{enumerateindent}
\end{theorem}

In order to obtain these results, we take the following steps:

\emph{Step 1:}
	First, we take a closer look at end components (ECs).
	Similar to the case of Markov chains, the maximal expected mean payoff with respect to the logarithmic rewards $\log(\reward)$ allows us to detect cases in which the values are $0$ or $\infty$.

\emph{Step 2:}
	In case the maximal expected logarithmic mean payoff in an EC is $0$, we distinguish whether the EC contains a BSCC in which all cycles are 1-cycles (\emph{non-gambling}) and whether the EC contains a mean payoff-optimal BSCC with n-cycles (\emph{gambling}).
	Then, we present an adapted version of the spider construction of \cite{BBDGS2018} allowing the iterative removal of all non-gambling BSCCs.

\emph{Step 3:}
	After pre-processing ECs by checking that no states  have value $\infty$ due to \emph{recurrent} behavior inside the ECs and iteratively removing non-gambling BSCCs, we provide a linear program (LP) whose solution contains the optimal values of all states.
	The correctness is shown by proving that the optimal values are the least fixed point of the update operator associated with the LP.
	If the LP has no solution, the value of the initial state is $\infty$ due to \emph{transient} behavior of the MDP.

\emph{Step 4:}
	An optimal MD-scheduler can be extracted from the solution to the LP and traced back to an MD-scheduler for the original MDP.
	The only case in which there is no optimal MD-scheduler is if the value $\ExpectationMDP<\MDP, \initialstate><\max>[\pathrewardinf] = \infty$ due to a gambling BSCC that can be left with arbitrarily high multiplicative reward before moving to a non-gambling BSCC (or an absorbing state with value $1$ as a special case).

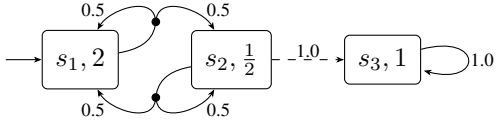
\begin{figure}
	\centering
	
	\begin{tikzpicture}[auto,initial text=]
		\node[state,initial=left] at (0,0) (s1) {$s_1, 2$\strut};
		\node[state] at (2,0) (s2) {$s_2, \tfrac{1}{2}$};
		\node[state] at (4,0) (s3) {$s_3, 1$};

		\node[actionnode] at (1, 0.5) (a11) {};
		\node[actionnode] at (1, -0.5) (a21) {};

		\path[->,probedge]
			(a11) edge[out=80,in=120,pos=0.8] node[prob] {0.5} (s2)
			(a11) edge[out=100,in=60,pos=0.8,swap] node[prob] {0.5} (s1)
			(a21) edge[out=-100,in=-60,pos=0.8] node[prob] {0.5} (s1)
			(a21) edge[out=-80,in=-120,pos=0.8,swap] node[prob] {0.5} (s2)
		;
		\path[-,actionedge]
			(s1) edge[out=10,in=-100] (a11)
			(s2) edge[out=190,in=80] (a21)
		;
		\path[directedge]
			(s2) edge[dashed] node[prob] {1.0} (s3)
			(s3) edge[loop right] node[prob] {1.0} (s3)
		;
	\end{tikzpicture}
	\caption{
		MDP to illustrate the special case of $\pathrewardinf$.
	} \label{fig:inf_special}
\end{figure}
\begin{remark} \label{rem:infimum_gambling_infinite}
	Before we continue, we illustrate the latter case, i.e.\ the infimum being $\infty$ due to a gambling BSCC with an exit.
	Consider the MDP in \cref{fig:inf_special}, first without the dashed action (equalling the right MC from \cref{fig:logzeromp_example}).
	As seen before, due to the \enquote{random walk without drift}-nature of this system, the infimum of the multiplicative reward is 0 for almost all paths.
	However, once the dashed action is added, a scheduler can at any time choose to \enquote{secure} the rewards of the current path.
	Now, the random walk actually plays in favour of the scheduler:
	By remembering the current rewards, the scheduler can wait until an arbitrary threshold is reached (which eventually happens almost surely) and then take the exiting action.
	As such, the supremum over all schedulers is infinite.
	We later also show that there exists a concrete (infinite memory) witness and that (as expected) no memoryless scheduler can obtain infinite reward here.
\end{remark}

All proofs for this section can be found in \Cref{appendix:proofsmdp}.

\subsection{Analysis and pre-processing of ECs}
\label{sec:ECs}

Starting with an arbitrary MDP, we first analyze the ECs by looking at all maximal ECs (MECs), i.e.\ ECs not contained in another EC, separately.
Each MEC is a strongly connected sub-MDP.
So, to ease notation, we  take a single strongly connected MDP $\MDP = (\States, \Actions, \mdptransitions)$ with reward function $\reward \colon \States \to \RealsPositive$ as given in the sequel.

\subsubsection{Logarithmic mean payoff}

Analogously to \Cref{lem:meanpayoffMC} for Markov chains, we obtain the following result:
\begin{restatable}{lemma}{mpMDP} \label{lem:MP_MDP}
	Let $\MDP$ and $\reward$ be as above and let $\initialstate\in \States$ be a state of $\MDP$.
	Then,
	1)~if $\ExpectationMDP<\MDP, \initialstate><\max>[\MP_{\log(\reward)}] < 0$, then $\ExpectationMDP<\MDP, \initialstate><\max>[\pathreward] = 0$ and hence $\ExpectationMDP<\MDP, \initialstate><\max>[\pathrewardinf] = 0$, and
	2)~if $\ExpectationMDP<\MDP, \initialstate><\max>[\MP_{\log(\reward)}] > 0$, then $\ExpectationMDP<\MDP, \initialstate><\max>[\pathrewardinf] = \infty$ and hence $\ExpectationMDP<\MDP, \initialstate><\max>[\pathreward] = \infty$.
\end{restatable}

Of course, checking whether the maximal logarithmic mean payoff is greater, less, or equal to $0$ is at least as hard as in the case of Markov chains for which no unconditional polynomial-time algorithm is  known (see \cref{prop:hardnessMP}).
A polynomial-space upper bound, however, is straightforward:

\begin{restatable}{proposition}{checklogMP}
\label{prop:check_logMP0}
Given a strongly connected $\MDP$  as above and a rational reward function $\reward$, deciding whether $\ExpectationMDP<\MDP, \initialstate><\max>[\MP_{\log(\reward)}] \bowtie 0$ for $\bowtie {}\in \{<,=,>\}$ is in PSPACE.
\end{restatable}

\subsubsection{Gambling and non-gambling BSCCs}
If the maximal expected mean payoff with respect to $\log(\reward)$ is zero, a more careful analysis is necessary.
As before, we need to investigate the nature of cycles within the system.
However, contrary to Markov chains, we cannot simply analyse all potential cycles:
Only cycles that are induced by MD-schedulers optimal with respect to the logarithmic mean payoff may be taken into consideration.
Only these cycles may be chosen with positive frequency by schedulers optimal with respect to the logarithmic mean payoff.
If a scheduler takes any other cycle with positive frequency, the logarithmic mean payoff is negative and therefore the multiplicative reward almost surely zero.
For further details, see \cite[Sec.~6.1]{Kallenberg}.

We follow the classification of ECs in \cite{BBDGS2018}.
First, we define the mean payoff-maximizing actions for each state.
Since in a strongly connected MDP,  any action can be chosen arbitrarily often by an optimal scheduler for  prefix-independent objective such as mean payoff, the definition requires the restriction to MD-schedulers.
For $s \in \States$, let
\begin{align*}
	\Actions^{\max}(s) = \{\alpha \in \Actions(s) \mid \alpha \text{ is part of a BSCC induced by an } \\
	\text{MD-scheduler $\strategy$ with $\ExpectationMDP<\MDP, \initialstate><\strategy>[\MP_{\log(\reward)}] = \ExpectationMDP<\MDP, \initialstate><\max>[\MP_{\log(\reward)}]$}\} .
\end{align*}
In other words, $\Actions^{\max}$ are all actions that an optimal scheduler may choose with positive frequency (a slightly stronger requirement than \enquote{infinitely often}).
Note that $\Actions^{\max}(s)$ is empty for states that are not part of any BSCC induced by a mean payoff-optimal MD-scheduler.
By $\MDP^{\max}$, we denote the sub-MDP that is obtained from $\MDP$ by disabling  all actions in $\Actions(s)\setminus \Actions^{\max}(s)$ for all states $s$.
We call ECs of $\MDP^{\max}$ \emph{max-ECs}, and BSCCs  of $\MDP^{\max}$ \emph{max-BSCCs}.
Note that there may be exponentially many such \emph{max-ECs} / \emph{max-BSCCs}.
As a by-product of checking whether the maximal logarithmic mean payoff is zero, we can compute the sets $\Actions^{\max}(s)$:
\begin{restatable}{proposition}{Amax}
Given a strongly connected $\MDP$ as above and a rational reward function $\reward\colon \States \to \Rationals_{>0}$ with  $\ExpectationMDP<\MDP, \initialstate><\max>[\MP_{\log(\reward)}]=0$ for any state $\initialstate\in \States$, the set $\Actions^{\max}(s)$ for all states $s\in \States$
can be computed in polynomial space.
\end{restatable}

\begin{remark}
Alternatively, the sets $\Actions^{\max}(s)$ can be characterized via the linear program for the optimal mean payoff in strongly connected MDPs (see \cite[Sec.~6.1]{Kallenberg}).
Using variables $v_s$ for all $s\in \States$ and a variable $x$, the LP is to minimize $x$ subject to
$x+ v_s \geq \log(r(s)) + {\sum}_{t\in \States} \mdptransitions(s,\alpha,t) \cdot v_t \quad \forall s\in \States, \alpha\in \Actions(s)$.
The optimal mean payoff is the value of $x$. 
For an optimal solution $(x^\ast,(v^\ast_s)_{s\in\States})$, we have $\alpha \in \Actions^{\max}(s)$ exactly when
the corresponding constraint is \emph{active} (i.e.\ left-hand \emph{equals} right-hand side).
However, due to the occurring logarithms, this representation is not directly applicable in an algorithm computing the sets $\Actions^{\max}(s)$.
\end{remark}

Assuming $\ExpectationMDP<\MDP, \initialstate><\max>[\MP_{\log(\reward)}] = 0$,  all max-ECs $E$ satisfy  $\ExpectationMDP<E, s><\sched>[\MP_{\log(\reward)}] = 0$ for any state $s$ and scheduler $\sched$ in $E$.
Now, we distinguish \emph{gambling} and \emph{non-gambling} max-ECs:

\begin{definition}
With $\MDP$ and $r$ as above with $\ExpectationMDP<\MDP, \initialstate><\max>[\MP_{\log(\reward)}] = 0$, we call a max-EC $E$ \emph{gambling}
if it contains an n-cycle (see \Cref{def:n-cycle}), and \emph{non-gambling} otherwise. 
Gambling and non-gambling max-BSCCs are such max-ECs with one action per state.
\end{definition}
The following lemma now states a result of \cite{BBDGS2018} for additive rewards translated to our setting.

\begin{restatable}{lemma}{gamblinglimsup}
\label{lem:gambling_limsup}
	Let $\MDP$ and $\reward$ be as above with $\ExpectationMDP<\MDP, \initialstate><\max>[\MP_{\log(\reward)}] = 0$.
	If there is a gambling max-BSCC, $\ExpectationMDP<\MDP, s><\max>[\pathreward] = \infty$  for all states $s \in \States$.
	Further, there is an MD-scheduler $\sched$ achieving $\ExpectationMDP<\MDP, s><\sched>[\pathreward] = \infty$ for all $s \in \States$.
\end{restatable}

Intuitively, staying in a gambling max-BSCC means that the sum of the logarithmic rewards on prefixes of a path follows a random walk without drift. So, arbitrarily high values are reached infinitely often almost surely.
If the sum of the  logarithmic rewards is arbitrarily high infinitely often, the same of course holds for the multiplicative reward of the prefixes. As $\pathreward$ is the $\limsup$ of these multiplicative rewards,
the obtained value almost surely equals $\infty$ in a gambling max-BSCC.

Analogously, staying in a gambling max-BSCC forever yields a $\pathrewardinf$ of $0$ almost surely.
So, without exiting the BSCC, \enquote{gambling} is not profitable for maximizing the expected value of $\pathrewardinf$.
If, however, a state $s$ with positive value can be reached from a gambling max-BSCC, the gambling max-BSCC can be used to almost surely generate an arbitrarily high multiplicative reward before moving to state $s$ with positive probability and realizing a positive value from then on, as sketched in \cref{rem:infimum_gambling_infinite}.

Note that we have just seen that (contrary to the case of Markov chains) the treatment of $\pathreward$ and $\pathrewardinf$ is not symmetric:
While the maximal expected value of $\pathreward$ is $\infty$ as soon as there is \emph{any} cycle with multiplicative reward different from $1$ in a max-BSCC, the maximal expected value of $\pathrewardinf$ is $0$ if  there exists such a cycle in every max-BSCC. So, in the case of $\pathrewardinf$ a quantifier alternation is introduced.

\subsubsection{Identifying non-gambling ECs}
In the following, we discuss how to algorithmically identify non-gambling max-ECs, i.e.\ max-ECs with only 1-cycles.
Before we proceed, observe that without the restriction to MDPs with maximal expected logarithmic mean-payoff zero, this problem is NP-complete (similar to the additive case \cite[Thm.~3.12]{BBDGS2018}).
The proof uses a reduction from the \emph{subset product} problem, similar to the well-known subset sum problem.
\begin{restatable}{lemma}{NPhardgambling}
	Deciding for a strongly connected MDP and reward function whether there exists an MD-scheduler such that all cycles under this scheduler are 1-cycles is NP-complete.
	This already holds if all distributions are Dirac (i.e.\ the MDP is a labelled transition system).
\end{restatable}

As it turns out, the restriction to $\Actions^{\max}$ in MDPs with maximal expected logarithmic mean-payoff zero reduces the complexity of the problem.
In particular, we show that individual actions can be pinpointed as \enquote{culprits} for n-cycles, instead of having to consider all possible combinations.
Without loss of generality, we consider an MDP $\MDP = (\States, \Actions, \mdptransitions)$ with reward function $\reward$ such that $\ExpectationMDP<\MDP, \initialstate><\max>[\MP_{\log(\reward)}] = 0$ that comprises one max-EC only containing logarithmic mean-payoff-optimal actions, i.e.\ it is one strongly connected sub-MDP of $\MDP^{\max}$.
Thus, in $\MDP$ \emph{every} scheduler $\sched$ has $\ExpectationMDP<\MDP, \initialstate><\sched>[\MP_{\log(\reward)}] = 0$ for any state $\initialstate$.
We want to determine whether this MDP contains a gambling BSCC or not.

\newcommand{\logtotalrewardto}[1]{\totalreward^{\to #1}_{\log(\reward)}}
\begin{restatable}{lemma}{rewardtounique} \label{lem:reward_to_is_unique}
	Let $s, t \in \States$ two states in $\MDP$, and $\sigma, \tau$ two MD-schedulers which almost surely reach $t$ from $s$, i.e.\ $\ProbabilityMDP<\MDP, s><\sigma>[\reach t] = 1$ and likewise for $\tau$.
	Then,
	\begin{equation*}
		\ExpectationMDP<\MDP, s><\sigma>[\logtotalrewardto{t}] = \ExpectationMDP<\MDP, s><\tau>[\logtotalrewardto{t}],
	\end{equation*}
	where $\logtotalrewardto{t}(\rho) = \sum_{i = 1}^{k} \log(\reward(\rho_i))$ and $k = \min \{i \mid \rho_i = t\} \union \{ \infty \}$.
	(For $k = \infty$, we consider w.l.o.g.\ the $\limsup$ of the sum -- since the set of paths that do not reach $t$ have measure $0$, their value is irrelevant.)
\end{restatable}

\begin{proof}[Proof sketch]
Let $\zeta$ be an MD-scheduler  that almost surely moves from $t$ back to $s$.
If $\ExpectationMDP<\MDP, s><\sigma>[\logtotalrewardto{t}] \not= \ExpectationMDP<\MDP, s><\tau>[\logtotalrewardto{t}]$, combining $\sigma$ and $\zeta$ or combining $\tau$ and $\zeta$
to one scheduler that moves from $s$ back to $s$ yields a scheduler that accumulates -- in expectation -- logarithmic reward different from $0$ between visits to $s$. This contradicts the fact that every scheduler has expected logarithmic reward zero.
\end{proof}

Thus, for any two states $s, t$ and scheduler $\strategy$ which almost surely moves from $s$ to $t$, the value $\ExpectationMDP<\MDP, s><\strategy>[\logtotalrewardto{t}]$ is the same and we refer to it with $Q(s, t)$, in particular $Q(s, s) = 0$.
However, this does not directly tell us anything about the value of actual cycles under different schedulers.

\begin{restatable}{lemma}{idetifygoodactions} \label{stm:identify_good_actions}
	Fix an EC $E$ in $\MDP$ and a state $x$ in $E$.
	For every state $s \in \States$, let $\Actions^\ast(s)$ denote all actions $a \in \Actions(s)$ where $\abs{\{ \log(\reward(u)) + Q(u, x) \mid u \in \support(\mdptransitions(s, a)) \}} = 1$.
	Then, $E$ contains only 1-cycles if and only if it only uses actions of $\Actions^\ast$.
\end{restatable}

\begin{proof}[Proof sketch]
	\underline{If:}
	By induction, it follows that when restricted to actions in $\Actions^\ast$ \emph{every path} from a state $s$ to $x$ has the same \emph{multiplicative} reward, namely $\exp(Q(s, x))$.
	Otherwise, if two paths would have a different reward, one could follow them until the point where they separate and get a contradiction.
	This then directly extends to any pair of states, i.e.\ all paths using only $\Actions^\ast$ from $s$ to $t$ yield $\exp(Q(s, t))$, and, in particular, $Q(s, s) = 0$, so $\exp(Q(s, s)) = 1$.

	\underline{Only if:}
	If at some state $s$ we follow an action not in $\Actions^\ast$, we get two paths from $s$ to $x$ which necessarily yield more and less than $Q(s, x)$, respectively.
	Concatenated with a path back from $x$ to $s$, we get two cycles with different values.
	In particular, at least one of them is not a 1-cycle.
\end{proof}
This directly leads us to an algorithm for computing $A^\ast$.
\begin{restatable}{lemma}{computeAstar}
In the $\MDP$ as above with $\ExpectationMDP<\MDP, \initialstate><\max>[\MP_{\log(\reward)}] = 0$ and only consisting of one max-EC,
	the set $A^\ast(s)$ can be computed in polynomial time.
\end{restatable}
\begin{proof}[Proof Sketch]
	Fix an arbitrary state $x$ and MD-scheduler $\sched$ under which $x$ is reached almost surely from every state.
	Then, for every starting state $s$ compute the \emph{expected visiting time} $\text{evt}_s(u)$ \cite{DBLP:conf/tacas/MertensKQW24}, i.e.\ how often is state $u$ visited on average until $x$ is reached, by solving a linear equation system.
	The expected reward until reaching $x$ can be expressed as $\ExpectationMDP<\MDP, s><\sched>[\logtotalrewardto{x}] = \sum_{u \in \States} \text{evt}_s(u) \cdot \log(\reward(t))$.
	This value is, in general, irrational.
	However, to decide $A^\ast$, we only need to decide whether all successors have the \emph{same} value (\cref{stm:identify_good_actions}).
	This can be reduced to queries of the form \enquote{$\sum p_i \cdot \log(r_i) = 0$}, which in turn reduces to an ESRI instance (as in \cref{prop:logMP_MC}).
\end{proof}

\begin{remark}
	Our approach directly characterizes \emph{all} gambling and non-gambling ECs through the set $A^\ast$.
	In comparison, \cite[Sec.~3.3]{BBDGS2018} identifies a non-gambling EC for the corresponding notion for additive rewards by iteratively picking candidates and modifying the transition probabilities of the MDP if that EC is gambling.
	For additive rewards, this works in polynomial time.
	Transferring this approach to multiplicative rewards does not directly yield a polynomial-time algorithm.
	In the other direction, we  believe that our characterization can be directly transferred to the additive case to provide an alternative polynomial-time algorithm that is arguably conceptually simpler.
\end{remark}

\subsubsection{Adapted spider construction}

\begin{figure} 
\centering
\tikzset{
	reward/.style={gray,font=\scriptsize\sffamily,outer sep=0pt,inner sep=1pt,anchor=center,draw,fill=white,rounded corners=1pt}
}
\begin{tikzpicture}[auto,xscale=1.1]
\node[state] at (0,0) (sinit) {$\initialstate, 1$};
\node[state,fill=white] at (0,-1.25) (s) {$s, \frac{1}{8}$\strut};
\node[state,fill=white] at (1.25,-1.25) (t) {$t, 2$\strut};
\node[state,fill=white] at (2.5,-1.25) (u) {$u, 4$\strut};

\node[state] at (0,-3) (t1) {$t_1$};
\node[state] at (1.25,-3) (t2) {$t_2$};
\node[state] at (2.5,-3) (t3) {$t_3$};

\path[directedge]
	(sinit) edge (s)
	(s) edge (t)
	(t) edge (u)
	(u) edge[bend right] (s)
	(t) edge (t2)
;

\path[probedge]
	(s) edge[out=-90,in=90,swap] node[prob] {$\frac{1}{2}$} (t1)
	(s) edge[out=-90,in=90] node[prob] {$\frac{1}{2}$} (t2)
	(u) edge[out=-90,in=90,swap] node[prob] {$\frac{2}{3}$} (t2)
	(u) edge[out=-90,in=90] node[prob] {$\frac{1}{3}$} (t3)
;
\begin{scope}[on background layer]
	\node[fit=(s) (u),rectangle,rounded corners,draw=lightgray,inner sep=4pt,thick] {};
\end{scope}
\end{tikzpicture}\quad
\begin{tikzpicture}[auto,xscale=1.1]
\node[state] at (0,0) (sinit) {$\initialstate, 1$};
\node[state,fill=white] at (0,-1.25) (s) {$s, \frac{1}{8}$\strut};
\node[state,fill=white] at (1.25,-1.25) (t) {$t, 2$\strut};
\node[state,fill=white] at (2.5,-1.25) (u) {$u, 4$\strut};
\node[state] at (1.25, 0) (bot) {$\bot, 1$};

\node[actionnode] at (0,-2.15) (se) {};
\node[actionnode] at (2.5,-2.15) (ue) {};

\node[state] at (0,-3) (t1) {$t_1$};
\node[state] at (1.25,-3) (t2) {$t_2$};
\node[state] at (2.5,-3) (t3) {$t_3$};

\path[directedge]
	(sinit) edge (s)
	(s) edge node[reward] {$1$} (t)
	(u) edge[swap] node[reward] {$\frac{1}{2}$} (t)
	(t) edge node[reward] {$1$} (t2)
	(t) edge node[reward] {$4$} (bot)
	(t) edge (ue)
;
\path[actionedge]
	(t) edge node[reward] {$\frac12$} (se)
	(t) edge[swap] node[reward] {$1$} (ue)
;

\path[probedge]
	(se) edge[out=-90,in=90,swap] node[prob] {$\frac{1}{2}$} (t1)
	(se) edge[out=-90,in=90] node[prob] {$\frac{1}{2}$} (t2)
	(ue) edge[out=-90,in=90,swap] node[prob] {$\frac{2}{3}$} (t2)
	(ue) edge[out=-90,in=90] node[prob] {$\frac{1}{3}$} (t3)
;
\begin{scope}[on background layer]
	\node[fit=(s) (u),rectangle,rounded corners,draw=lightgray,inner sep=4pt,thick] {};
\end{scope}
\end{tikzpicture}
\caption{
	Illustration of the adapted spider construction for $\pathreward$.
	The gray end component in the left MDP is removed on the right.
	All transitions leaving the end component are moved to start from state $t$ after the construction.
	The weights are adjusted accordingly.
	Additionally, one can move to the new state $\bot$ from $t$.
	As in the text, we write rewards on transitions for readability.
}
\label{fig:spider}
\end{figure}
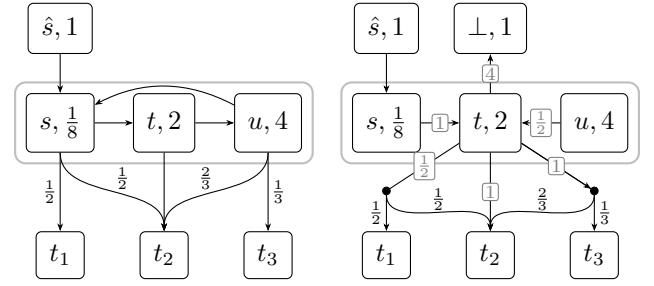

Let $\MDP = (\States, \Actions, \mdptransitions)$ with reward function $r$ and initial state $\initialstate \in \States$ be a strongly connected MDP as above. 
Suppose $\ExpectationMDP<\MDP, \initialstate><\max>[\MP_{\log(\reward)}] = 0$.

We provide an adaptation of the \emph{spider construction} of \cite{BBDGS2018} to remove non-gambling max-BSCCs $E$ (containing only cycles with multiplicative reward $1$).
So, let $E = (T, B)$ be such a max-BSCC where $B \colon T \to \Actions$ specifies which action is part of the BSCC.
We construct an MDP $\MDP_E = (\States', \Actions', \mdptransitions')$ in which the non-gambling max-BSCC $E$ is removed.

Since all cycles have multiplicative reward $1$, for any pair of states $s, t$ there again is a unique multiplicative reward for all paths from $s$ to $t$.
Define this value as $R(s, t)$.
Note that $R(s, t)$ can be computed exactly as in \cref{sec:mc:scc}.

Let $s_E$ be an arbitrary fixed state in $E$.
In the sequel, to ease the formulation, we use the following abbreviation:
When we say we add an action to state $s$ with successor distribution $d$ and reward $x$, this means creating an intermediate state $s_x$ with reward $x$, to which $s$ can transition with probability $1$, and adding a single action to $s_x$ which has the successor distribution $d$.
We then perform the following steps:

Step~1:
	First, we compute the value of the state $s_E$ within the end component.
	Here, we have to distinguish whether we are interested in $\pathreward$  or in $\pathrewardinf$.

Step~2:
	If a scheduler stays in $E$ forever, the  expected value of $\pathreward$ when starting in $s_E$ is $\overline{v_E} \coloneqq \max_{t\in T} R(s_E,t)$.
	Now, we add a new absorbing state $\bot$ with a self-loop and reward $1$.
	In $s_E$, we enable a new action $\mathit{stay}$ that leads to $\bot$ with probability $1$ and reward $\overline{v_E}/\reward(s_E)$.

	Regarding $\pathrewardinf$, the  expected value of $\pathrewardinf$ when starting in $s_E$ is $\underline{v_E} \coloneqq \min_{t\in T} R(s_E,t)$.
	Also in this case, we add the absorbing state $\bot$ with a self-loop and reward $1$ and the  action $\mathit{stay}$ that leads to $\bot$ with probability $1$ and reward $\underline{v_E}/\reward(s_E)$.

Step~3:
	Next, we remove the action  $B(t)$ from all states $t \in T \setminus \{s_E\}$.
	To replace these, we add the action $\mathit{center}_t$ to each $t$, which leads to the state $s_E$ with probability $1$ and with reward $R(t, s_E)/\reward(t)$.

Step~4:
	Finally, for each state $t \in T$ and  $\beta \not = B(t)$, we enable a new action $\tau_{t, \beta}$ in state $s_E$ leading to the same successors as $\beta$ in $t$ with the same probabilities.
	More formally, for all $s\in S$, we let $\mdptransitions'(s_E, \tau_{t,\beta}, s) = \mdptransitions(t, \beta, s)$ and assign a reward of $R(s_E,t) \cdot \reward(t)/\reward(s_E)$ to $\tau_{t, \beta}$.
All remaining transition probabilities and rewards remain unchanged.
Denote the resulting MDP by $\MDP_E=(\States', \Actions', \mdptransitions')$.

\begin{restatable}{lemma}{correctnessspider} \label{stm:spider_preserves_value}
	We have $\ExpectationMDP<\MDP, \initialstate><\max>[\pathreward] = \ExpectationMDP<\MDP_E, \initialstate><\max>[\pathreward]$ when using the value $\overline{v_E}$ in the adapted spider construction.
	Analogously, $\ExpectationMDP<\MDP, \initialstate><\max>[\pathrewardinf] = \ExpectationMDP<\MDP_E, \initialstate><\max>[\pathrewardinf]$ with the value $\underline{v_E}$.
\end{restatable}

\begin{remark}
The correctness of the value $\underline{v_E}$ used for action $\mathit{stay}$ in case we investigate $\pathrewardinf$ relies on the fact that $E$ is a BSCC and hence staying in $E$ almost surely 
implies visiting every state in $E$ infinitely often. Applying the adapted spider construction to ECs instead of BSCCs would require additional work to determine the correct value to use for the 
action $\mathit{stay}$. 
For $\pathreward$, on the other hand, we could apply the adapted spider construction also to non-gambling max-ECs that are not BSCCs.
The reason is that staying in an end component $E$, we would have $\ExpectationMDP<E, t><\max>[\pathreward] = \max_{t'\in T} R(t,t')$. In contrast,
$\ExpectationMDP<E, t><\max>[\pathrewardinf] \not= \min_{t'\in T} R(t,t')$ is possible.
\end{remark}

\subsubsection{Iterative application of the adapted spider construction}
Still assuming that $\ExpectationMDP<\MDP, \initialstate><\max>[\MP_{\log(\reward)}] = 0$, we apply the adapted spider construction iteratively to remove all non-gambling max-BSCCs.
Starting with the MDP $\MDP_0=\MDP$, we construct a sequence of MDPs as follows: If $\MDP_i$ has been constructed and contains a non-gambling max-BSCC that is not an absorbing state,
we remove one such BSCC with the spider construction. While doing that, we always choose states present in the original MDP $\MDP$ as the center of the spider construction and not auxiliary states added in the process.  This is always possible: As all new auxiliary states introduced in previous applications of the spider construction can only be reached via one action from a state present before the respective application of the spider construction, an auxiliary state is only part of a BSCC if also the corresponding state from the original MDP is part of the BSCC.
 We let $\MDP_{i+1}$ be the resulting MDP. If $\MDP_i$ contains no non-gambling max-BSCC except for absorbing states, 
the process terminates.
To show that this always happens,  
we use the following result.

\begin{restatable}{lemma}{spidernumber}
\label{lem:spider_number_iterations}
With $\MDP = (\States, \Actions, \mdptransitions)$ and $\MDP_E = (\States', \Actions', \mdptransitions')$ as above, we have
$
{\sum}_{t\in S} \left|\Actions(t)\right| - 1  > {\sum}_{t\in S'} \left|\Actions'(t)\setminus\{\mathit{stay}\}\right| - 1
$.
\end{restatable}
This allows us to bound the necessary  number of applications of the spider construction:

\begin{restatable}{lemma}{boundspider}
The iterative removal of non-gambling max-BSCCs terminates after at most $\sum_{t\in S} \left|\Actions(t)\right|$ successive applications of the adapted spider construction.
\end{restatable}

All in all, we obtain:
\begin{lemma}
\label{thm:poly_spider_application}
\label{lem:complexity_pre-processing}
In an MDP $\MDP$ with rational reward function $\reward$ as above, the iterative identification and removal of non-gambling max-BSCCs can be done in polynomial space.
\end{lemma}

\subsection{Characterization and computation of maximal expected multiplicative rewards in MDPs}

Given an arbitrary MDP $\MDP=(\States,A,\mdptransitions)$, an initial state $\initialstate\in \States$ and a reward function $\reward \colon \States \to \RealsPositive$, we now address the problem to find the values $\ExpectationMDP<\MDP, \initialstate><\max>[\pathreward]$ and $\ExpectationMDP<\MDP, \initialstate><\max>[\pathrewardinf]$.
Recall that we assume that all states are reachable.

\subsubsection{Applying the analysis of ECs}
By the previous section, we know that both values are infinite if there is an end component $\cE$ with $\ExpectationMDP<\cE><\max>[\MP_{\log(\reward)}] > 0$.
This can be checked in polynomial space by \cref{prop:check_logMP0}.
If the maximal logarithmic mean-payoff is zero, we can furthermore remove all non-gambling max-BSCCs in polynomial space by \Cref{lem:complexity_pre-processing}.

\begin{assumption} \label{ass:log-mp}
	All end components satisfy $\ExpectationMDP<\cE><\max>[\MP_{\log(\reward)}] \leq 0$ and the spider construction has been applied to all ECs that only contain 1-cycles, i.e.\ there are no such ECs except for absorbing states with reward $1$.\footnote{Recall here that the reward of new intermediate states added by the spider construction depends on whether we investigate $\pathreward$ or $\pathrewardinf$.}
\end{assumption}
Regarding $\ExpectationMDP<\MDP, \initialstate><\max>[\pathreward]$, \Cref{lem:gambling_limsup} tells us that the existence of an end component $\cE$ with $\ExpectationMDP<\cE><\max>[\MP_{\log(\reward)}] = 0$ now implies that $\ExpectationMDP<\MDP, \initialstate><\max>[\pathreward]=\infty$
as there are gambling max-BSCCs in this case by our assumption.
So, when addressing $\ExpectationMDP<\MDP, \initialstate><\max>[\pathreward]$, we will assume that all end components $\cE$ except for absorbing states satisfy $\ExpectationMDP<\cE><\max>[\MP_{\log(\reward)}] < 0$.

For  $\ExpectationMDP<\MDP, \initialstate><\max>[\pathrewardinf]$, we have to take a closer look at remaining end components $\cE$ with $\ExpectationMDP<\cE><\max>[\MP_{\log(\reward)}] = 0$.

\begin{restatable}{lemma}{ECliminf}
\label{lem:ECliminf}
Let $\MDP$ be an MDP as above satisfying \Cref{ass:log-mp}. For a state $s$ in an EC $\cE$ with $\ExpectationMDP<\cE><\max>[\MP_{\log(\reward)}] = 0$, we have that $\ExpectationMDP<\MDP, \initialstate><\max>[\pathrewardinf]=\infty$ if and only if an absorbing state with reward $1$ is reachable from $s$, and $\ExpectationMDP<\MDP, \initialstate><\max>[\pathrewardinf]=0$ otherwise (i.e.\ no absorbing state with reward $1$ is reachable).
\end{restatable}

\begin{proof}[Proof sketch]
If an absorbing state with reward $1$ is reachable, a scheduler can stay in a gambling max-BSCC in $\cE$ until arbitrarily high multiplicative rewards are obtained, which happens almost surely, and then move to the absorbing state. 
Otherwise, all reachable ECs are gambling or have negative maximal  expected logarithmic mean payoff. In both cases, staying in such an EC leads to value $0$.
\end{proof}
If the situation of the lemma leading to value $\infty$ is the case, even a single infinite-memory scheduler can achieve value $\infty$:
The scheduler makes sure that a multiplicative reward of $4^i$ is achieved with probability at least $(1/2)^i$ (at least for sufficiently large $i$).

Algorithmically, the condition boils down to a simple reachability check given the results on the computation of max-ECs and non-gambling max-BSCCs presented above.
If the value of some state is $\infty$, also $\ExpectationMDP<\MDP, \initialstate><\max>[\pathrewardinf]=\infty$ as we assume that all states are reachable. 

If we assume that no state has value $\infty$ due to \Cref{lem:ECliminf},  all states in gambling max-BSCCs have value $0$. We can collapse all these states with value $0$ to a single absorbing state with value $0$.
Then, all remaining end components $\cE$ satisfy $\ExpectationMDP<\cE><\max>[\MP_{\log(\reward)}] < 0$.
So, also when addressing $\ExpectationMDP<\MDP, \initialstate><\max>[\pathrewardinf]$, we will from now on assume the following:

\begin{assumption}
\label{ass:negativeMP}
Let $\MDP$, $\initialstate$ and $\reward$ be as above.
We assume that all end components $\cE$ except for absorbing states satisfy $\ExpectationMDP<\cE><\max>[\MP_{\log(\reward)}] < 0$.
Further, we assume that all absorbing states have reward $0$ or $1$.
\end{assumption}

In summary, with our pre-processing we can check whether maximal expected values of $\pathreward$ or $\pathrewardinf$ are $\infty$ due to the behavior inside ECs. If this is not the case,
we can transform the MDP in polynomial space such that \Cref{ass:negativeMP} holds.

\subsubsection{Linear program for $\ExpectationMDP<\MDP, \initialstate><\max>[\pathreward]$ and $\ExpectationMDP<\MDP, \initialstate><\max>[\pathrewardinf]$}
\label{sec:LP}
Given an arbitrary MDP $\MDP=(\States,A,\mdptransitions)$, an initial state $\initialstate\in \States$ and a reward function $\reward \colon \States \to \RealsPositive$, we work under \Cref{ass:negativeMP}.
So, we assume that all states are reachable from $\initialstate$ and that all end components $\cE$ satisfy $\ExpectationMDP<\cE><\max>[\MP_{\log(\reward)}] < 0$ except for absorbing states with reward $1$.

\begin{figure}[t]
\begin{align*}
	&\text{Minimize ${\sum}_{s\in \States} x_s$ subject to}\tag{LP1}\label{eq:mdp_lp} \\
	&x_s \geq 0, \text{ $\forall s\in \States$,} \qquad x_t  = \reward(t), \text{ $\forall t \in T$,} \qquad \text{and} \\
	&x_s \geq  \reward(s) \cdot {\sum}_{s' \in \States} \mdptransitions(s,\alpha,s') \cdot x_{s'}, \text{ $\forall s\in S\setminus T$, $\alpha \in A(s)$.}
\end{align*}\vspace{-.5cm}
\caption{The linear program to optimize transient behaviour, called \ref{eq:mdp_lp}.} \label{fig:mdp_lp}
\end{figure}

Now, we construct a linear program with one variable $x_s$ per state $s\in \States$.
Let $T\subseteq \States$ be the set of absorbing states.
By \Cref{ass:negativeMP}, these states all have reward $0$ or $1$.
The linear program is given in \cref{fig:mdp_lp} (referred to as \ref{eq:mdp_lp}).
We show that solutions to it correspond to the values we seek.

\begin{theorem}
Consider the MDP $\MDP$, initial state $\initialstate$, and the reward function $\reward$ satisfying \Cref{ass:negativeMP}.
Then, the values
$v_s = \ExpectationMDP<\MDP, s><\max>[\pathreward] = \ExpectationMDP<\MDP, s><\max>[\pathrewardinf]$ for $s\in \States$
are the unique solution to \ref{eq:mdp_lp} if a finite solution exists.
If no feasible point exists, then $v_{\initialstate}=\infty$.
\end{theorem}

We prove the theorem step by step in the following lemmata.
We begin by establishing that the  non-feasibility of \ref{eq:mdp_lp} implies that the value is infinite.
\begin{restatable}{lemma}{feasabilityLP}
Let $\MDP$, $\initialstate$, and $\reward$ be as above.
If \ref{eq:mdp_lp} has no (finite) solution, then $\ExpectationMDP<\MDP, \initialstate><\max>[\pathreward] = \ExpectationMDP<\MDP, \initialstate><\max>[\pathrewardinf] = \infty$.
\end{restatable}

Next, assuming that LP is feasible, we show that the solution is the least fixed point of an update operator corresponding to a Bellmann optimality equation.
This update operator $T\colon \Reals_{\geq 0}^\States  \to \Reals_{\geq 0}^\States $ is defined by $T_s(x) = \reward(s)$ if $s\in T$, and ${\max}_{\alpha\in A(s)} \reward(s) \cdot {\sum}_{s' \in \States} \mdptransitions(s,\alpha,s') \cdot x_{s'}$ otherwise, i.e.\ $s\in \States\setminus T$.
for all $s\in \States$ and $x\in  \Reals_{>0}^\States$.
This operator is monotone with respect to the pointwise ordering $\geq$.
That is, if $y\geq x$ holds componentwise, then $T(y) \geq T(x)$ holds componentwise.
Using the Knaster-Tarski theorem, we show:

\begin{restatable}{lemma}{knastertarski}
Let $\MDP$, $\initialstate$, $\reward$, and the function $T$ be as above.
Let $x\in  \Reals_{>0}^\States$ be a solution to \ref{eq:mdp_lp}.
Then, $x$ is the least fixed point of $T$.
\end{restatable}

We define the random variable $\multreward_n$ on paths $\infinitepath$ as $\multreward_n(\infinitepath) \coloneqq \prod_{i\leq n} \reward(\infinitepath_i)$ if $s_n\in T$ and $0$ otherwise.
By induction, we get a direct connection between the expected value of $\multreward_n$ and $T^n$.

\begin{lemma}
Let $\MDP$, $\initialstate$, $\reward$, and the function $T$ be as above.
Then $T^{n+1}(0) = \left( \ExpectationMDP<\MDP, s><\max>[\multreward_n] \right)_{s\in \States}$.
\end{lemma}
Conditioning on whether a state in $T$ is reached or not to split up expected values, we can now conclude:

\begin{restatable}{lemma}{nstepconverges}
Let $\MDP$, $\initialstate$, and $\reward$ be as above.
Then, for all states $s\in \States$,
$
{\lim}_{n\to \infty}  \ExpectationMDP<\MDP, s><\max>[\multreward_n] = \ExpectationMDP<\MDP, s><\max>[\pathreward] = \ExpectationMDP<\MDP, s><\max>[\pathrewardinf]
$.
\end{restatable}
So, under \Cref{ass:negativeMP}, the optimal values can be computed by solving \ref{eq:mdp_lp} which can be constructed in polynomial time from an MDP with rational weight function.

\begin{theorem}
Let $\MDP=(\States,A,\mdptransitions)$ be  an  MDP,   $\initialstate\in \States$ an initial state, and $\reward \colon \States \to \Rationals_{>0}$  a reward function.
 Assume \Cref{ass:negativeMP} holds. Then, $\ExpectationMDP<\MDP, \initialstate><\max>[\pathreward]=\ExpectationMDP<\MDP, \initialstate><\max>[\pathrewardinf]$ can be computed in polynomial time.
\end{theorem}

As establishing \Cref{ass:negativeMP} can be done in polynomial space as shown in the previous section, we obtain the following complexity upper bound for the general problem to compute maximal expected multiplicative rewards:

\begin{theorem}
\label{thm:MDP_mult_rew}
	Let $\MDP=(\States,A,\mdptransitions)$ be  an  MDP, $\initialstate\in \States$, and $\reward \colon \States \to \Rationals_{>0}$ a reward function.
	Then, $\ExpectationMDP<\MDP, s><\max>[\pathreward]$ and $\ExpectationMDP<\MDP, s><\max>[\pathrewardinf]$ can be computed in polynomial space.
\end{theorem}
Naturally, the CSRI-hardness result for Markov chains (\Cref{prop:hardnessMP}) holds also for MDP.
Closing the resulting complexity gap remains as future work.

\begin{remark}
	States with reward $0$ can be handled similar as in MC.
	In particular, w.l.o.g.\ they can be made absorbing, and then be treated like a value-0 EC in the transient analysis.
\end{remark}

\subsection{Scheduler Complexity}

We prove that MD-schedulers are sufficient for the maximization in  MDPs using the representation of the maximal expected value of $\pathreward$  in MDPs satisfying \Cref{ass:negativeMP} as the solution to \ref{eq:mdp_lp}.
For $\pathreward$, we have seen that also divergence to $\infty$ in recurrent states (in MDPs not  satisfying \Cref{ass:negativeMP})   can be achieved by MD-schedulers in \Cref{lem:gambling_limsup}.

For $\pathrewardinf$, the situation is different: If the value is finite, an optimal MD-scheduler can be obtained as for $\pathreward$. If there are gambling max-ECs in MDPs with maximal logarithmic mean payoff zero,
we have already seen in \Cref{rem:infimum_gambling_infinite} that schedulers might have to ``gamble'' for a while before leaving the EC in order to obtain arbitrarily high values. This is the only situation in which MD-schedulers are not sufficient for the maximization of expected multiplicative rewards.

\begin{restatable}{theorem}{schedulercomplexity}
	Let $\MDP=(\States,A,\mdptransitions)$ be an MDP, $\initialstate\in \States$ an initial state, and  $\reward \colon \States \to \Reals_{>0}$  a reward function.
	There is an MD-scheduler $\sched$ with $\ExpectationMDP<\MDP, \initialstate><\sched>[\pathreward]=\ExpectationMDP<\MDP, \initialstate><\max>[\pathreward]$.
	
	If $\ExpectationMDP<\MDP, \initialstate><\max>[\pathrewardinf] < \infty$, there is an MD-scheduler $\sched$ with $\ExpectationMDP<\MDP, \initialstate><\sched>[\pathrewardinf]=\ExpectationMDP<\MDP, \initialstate><\max>[\pathrewardinf]$.
\end{restatable}

Finally, the result that MD-schedulers are always sufficient for the maximization of the expected value of $\pathreward$ allows us to conclude the following complexity bound for the threshold problem by guessing an MD-scheduler and verifying the result with a CSRI-oracle:

\begin{restatable}{corollary}{thresholdproblem}
Let $\MDP=(\States,A,\mdptransitions)$ be an MDP, $\initialstate\in \States$ an initial state, $\reward \colon \States \to \Rationals_{>0}$  a reward function, and $\vartheta\in \Rationals$ a threshold.
Deciding whether $\ExpectationMDP<\MDP, \initialstate><\sched>[\pathreward] \geq \vartheta$    is in $\textrm{NP}^{\textrm{CSRI}}$.
\end{restatable}

\subsection{Multiplicative Stochastic Shortest Path}

The classical \emph{stochastic shortest path problem} asks for a scheduler in an MDP maximizing (or minimizing) the expected total accumulated additive reward among all schedulers reaching a given target state with probability $1$.
This problem is solvable in polynomial time \cite{BerTsi91,deAlf99,BBDGS2018}.

Analogously, the \emph{multiplicative stochastic shortest path problem} can be formulated as follows:
Given an MDP $\MDP=(\States,A,\mdptransitions)$, an initial state $\initialstate \in \States$, a reward function $\reward \colon \States \to \Rationals_{>0}$, and an absorbing state $t \in \States$, compute $\sup_{\sched} \ExpectationMDP<\MDP, \initialstate><\sched>[\pathreward]$ where $\sched$ ranges over all schedulers with $\ProbabilityMDP<\MDP, \initialstate><\sched>[\lozenge t] = 1$.
Furthermore, compute an optimal scheduler if the optimal value is finite.
Note that $\ExpectationMDP<\MDP, \initialstate><\sched>[\pathreward] = \ExpectationMDP<\MDP, \initialstate><\sched>[\pathrewardinf]$ for all schedulers $\sched$ reaching $t$ almost surely.

The techniques developed in this section, specifically \cref{thm:MDP_mult_rew}, are directly applicable to solve this problem:
\begin{corollary}
	The multiplicative stochastic shortest path problem can be solved in polynomial space.
\end{corollary}

\begin{proof}[Proof sketch]
Given $\MDP=(\States,A,\mdptransitions)$, $\initialstate\in \States$,  $\reward \colon \States \to \Rationals_{>0}$, and  $t\in \States$ as above,
the multiplicative stochastic shortest path problem can be solved as follows:
First, all states $s$ with  $\ProbabilityMDP<\MDP, \initialstate><\max>[\lozenge t] < 1$ are removed.
Then, end components are analyzed as in \Cref{sec:ECs}.
If the maximal expected multiplicative reward in $\infty$ in some end component $E$, then also the value in the multiplicative stochastic shortest path problem is $\infty$:
A sequence of schedulers reaching $t$ almost surely with arbitrarily high expected multiplicative rewards can be obtained as follows:
Move to $E$ with positive probability $p_1$ along some fixed path $\finitepath_1$ with multiplicative reward $w_1$.
In case trying to reach $E$ via this path is not successful, move to $t$ with probability $1$.
If $E$ is reached via $\finitepath_1$, stay in $E$ until an arbitrarily high multiplicative reward of $k$ is obtained inside $E$, and then moving to $t$ with probability at least $p_2$ and with multiplicative reward at least $w_2$ via a fixed path $\finitepath_q$ depending only on the current state $q$ within $E$.
The resulting scheduler has an expected multiplicative reward of at least $p_1\cdot w_1\cdot k \cdot p_2\cdot w_2$, diverging to $\infty $ for $k\to\infty$.

If there is no end component in which infinite multiplicative reward can be obtained, remove all non-gambling BSCCs with the spider construction without adding the action $\mathit{stay}$ and new absorbing states.
The schedulers under consideration have to leave a non-gambling BSCC almost surely and staying inside such a BSCC for a while before leaving does not influence the multiplicative reward of the resulting path.

Now, all remaining end components have negative maximal expected mean payoff with respect to the logarithmic rewards.
So, the maximal expected multiplicative reward can be computed as in \Cref{sec:LP}.
The resulting scheduler will reach $t$ almost surely as it could otherwise be improved by redirecting paths not reaching $t$ towards $t$ if such paths had positive probability mass.
\end{proof}

\section{Conclusion}

In this work, we presented multiplicative rewards for Markovian system, namely Markov chains and MDPs.
We demonstrated several key differences to standard additive rewards, such as divergence of values potentially stemming from transient states.
We established complexity bounds for determining and optimizing expected multiplicative rewards in Markov chains and MDPs, respectively, and considered several special cases.
We showed that several of the complexity results reduce to polynomial time if one of two well-known number-theoretic conjectures hold.
Additionally, we provided precise scheduler complexity bounds for MDPs.

For theoretical extensions, we identify several routes.
Firstly, an extension to games is interesting, especially considering how the additional quantifier alternation in the case of $\pathrewardinf$ already introduced asymmetries.
Next, we conjecture that instead of focussing on multiplication, some results could be extended to more general structures (e.g.\ semi-rings).
Practically motivated, we believe that a better algorithm than  iterating through all MD-schedulers could be found for determining the  nature of the logarithmic mean payoff in strongly connected MDP (\cref{prop:check_logMP0}).
With this, a practical algorithm, potentially making use of approximations, could be developed.
Finally, our framework can be applied to optimizing entropic risk and remove restrictions of \cite{DBLP:journals/iandc/BaierCMP24}.

\newpage
\bibliography{main}

\begin{thebibliography}{10}

\bibitem{DBLP:journals/mst/AkianGGG19}
Marianne Akian, St{\'{e}}phane Gaubert, Julien Grand{-}Cl{\'{e}}ment, and
  J{\'{e}}r{\'{e}}mie Guillaud.
\newblock The operator approach to entropy games.
\newblock {\em Theory Comput. Syst.}, 63(5):1089--1130, 2019.
\newblock \href {https://doi.org/10.1007/S00224-019-09925-Z}
  {\path{doi:10.1007/S00224-019-09925-Z}}.

\bibitem{DBLP:conf/icalp/AllamigeonGKS22}
Xavier Allamigeon, St{\'{e}}phane Gaubert, Ricardo~D. Katz, and Mateusz Skomra.
\newblock Universal complexity bounds based on value iteration and application
  to entropy games.
\newblock In Mikolaj Bojanczyk, Emanuela Merelli, and David~P. Woodruff,
  editors, {\em 49th International Colloquium on Automata, Languages, and
  Programming, {ICALP} 2022, July 4-8, 2022, Paris, France}, volume 229 of {\em
  LIPIcs}, pages 110:1--110:20. Schloss Dagstuhl - Leibniz-Zentrum f{\"{u}}r
  Informatik, 2022.
\newblock \href {https://doi.org/10.4230/LIPICS.ICALP.2022.110}
  {\path{doi:10.4230/LIPICS.ICALP.2022.110}}.

\bibitem{DBLP:conf/stacs/AsarinCDDHK16}
Eugene Asarin, Julien Cervelle, Aldric Degorre, Catalin Dima, Florian Horn, and
  Victor~S. Kozyakin.
\newblock Entropy games and matrix multiplication games.
\newblock In Nicolas Ollinger and Heribert Vollmer, editors, {\em 33rd
  Symposium on Theoretical Aspects of Computer Science, {STACS} 2016, February
  17-20, 2016, Orl{\'{e}}ans, France}, volume~47 of {\em LIPIcs}, pages
  11:1--11:14. Schloss Dagstuhl - Leibniz-Zentrum f{\"{u}}r Informatik, 2016.
\newblock \href {https://doi.org/10.4230/LIPICS.STACS.2016.11}
  {\path{doi:10.4230/LIPICS.STACS.2016.11}}.

\bibitem{BBDGS2018}
Christel Baier, Nathalie Bertrand, Clemens Dubslaff, Daniel Gburek, and Ocan
  Sankur.
\newblock Stochastic shortest paths and weight-bounded properties in {M}arkov
  decision processes.
\newblock In Anuj Dawar and Erich Gr{\"{a}}del, editors, {\em Proceedings of
  the 33rd Annual {ACM/IEEE} Symposium on Logic in Computer Science, {LICS}
  2018, Oxford, UK, July 09-12, 2018}, pages 86--94. {ACM}, 2018.

\bibitem{DBLP:journals/iandc/BaierCMP24}
Christel Baier, Krishnendu Chatterjee, Tobias Meggendorfer, and Jakob
  Piribauer.
\newblock Entropic risk for turn-based stochastic games.
\newblock {\em Inf. Comput.}, 301:105214, 2024.
\newblock \href {https://doi.org/10.1016/J.IC.2024.105214}
  {\path{doi:10.1016/J.IC.2024.105214}}.

\bibitem{DBLP:books/daglib/0020348}
Christel Baier and Joost{-}Pieter Katoen.
\newblock {\em Principles of model checking}.
\newblock {MIT} Press, 2008.

\bibitem{Baker+1998+37+44}
A.~Baker.
\newblock Logarithmic forms and the abc-conjecture.
\newblock In Kálmán Győry, Attila Pethő, and Vera~T. Sós, editors, {\em
  Diophantine, Computational and Algebraic Aspects. Proceedings of the
  International Conference held in Eger, Hungary, July 29-August 2, 1996},
  pages 37--44, Berlin, New York, 1998. De Gruyter.
\newblock \href {https://doi.org/doi:10.1515/9783110809794.37}
  {\path{doi:doi:10.1515/9783110809794.37}}.

\bibitem{BerTsi91}
Dimitri~P. Bertsekas and John~N. Tsitsiklis.
\newblock An analysis of stochastic shortest path problems.
\newblock {\em Mathematics of Operations Research}, 16(3):580--595, 1991.

\bibitem{brandtner2018entropic}
Mario Brandtner, Wolfgang K{\"u}rsten, and Robert Rischau.
\newblock Entropic risk measures and their comparative statics in portfolio
  selection: Coherence vs. convexity.
\newblock {\em European Journal of Operational Research}, 264(2):707--716,
  2018.

\bibitem{DBLP:conf/stoc/Canny88}
John~F. Canny.
\newblock Some algebraic and geometric computations in {PSPACE}.
\newblock In Janos Simon, editor, {\em Proceedings of the 20th Annual {ACM}
  Symposium on Theory of Computing, May 2-4, 1988, Chicago, Illinois, {USA}},
  pages 460--467. {ACM}, 1988.
\newblock \href {https://doi.org/10.1145/62212.62257}
  {\path{doi:10.1145/62212.62257}}.

\bibitem{deAlf99}
Luca de~Alfaro.
\newblock Computing minimum and maximum reachability times in probabilistic
  systems.
\newblock In {\em 10th International Conference on Concurrency Theory
  (CONCUR)}, volume 1664 of {\em Lecture Notes in Computer Science}, pages
  66--81, 1999.

\bibitem{di1999risk}
Giovanni~B Di~Masi and Lukasz Stettner.
\newblock Risk-sensitive control of discrete-time {M}arkov processes with
  infinite horizon.
\newblock {\em SIAM Journal on Control and Optimization}, 38(1):61--78, 1999.

\bibitem{DBLP:journals/toct/EtessamiSY14}
Kousha Etessami, Alistair Stewart, and Mihalis Yannakakis.
\newblock A note on the complexity of comparing succinctly represented
  integers, with an application to maximum probability parsing.
\newblock {\em {ACM} Trans. Comput. Theory}, 6(2):9:1--9:23, 2014.
\newblock \href {https://doi.org/10.1145/2601327} {\path{doi:10.1145/2601327}}.

\bibitem{follmer2011entropic}
Hans F{\"o}llmer and Thomas Knispel.
\newblock Entropic risk measures: Coherence vs. convexity, model ambiguity and
  robust large deviations.
\newblock {\em Stochastics and Dynamics}, 11(02n03):333--351, 2011.

\bibitem{follmer2002convex}
Hans F{\"o}llmer and Alexander Schied.
\newblock Convex measures of risk and trading constraints.
\newblock {\em Finance and stochastics}, 6(4):429--447, 2002.

\bibitem{gillespie1976general}
Daniel~T Gillespie.
\newblock A general method for numerically simulating the stochastic time
  evolution of coupled chemical reactions.
\newblock {\em Journal of Computational Physics}, 22(4):403--434, 1976.

\bibitem{G_mez_2010}
S.~G{\'{o}}mez, A.~Arenas, J.~Borge-Holthoefer, S.~Meloni, and Y.~Moreno.
\newblock Discrete-time {M}arkov chain approach to contact-based disease
  spreading in complex networks.
\newblock {\em {EPL} (Europhysics Letters)}, 89(3), Feb 2010.

\bibitem{howard1972risk}
Ronald~A Howard and James~E Matheson.
\newblock Risk-sensitive {M}arkov decision processes.
\newblock {\em Management science}, 18(7):356--369, 1972.

\bibitem{jaquette1976utility}
Stratton~C Jaquette.
\newblock A utility criterion for {M}arkov decision processes.
\newblock {\em Management Science}, 23(1):43--49, 1976.

\bibitem{kaelbling1996reinforcement}
Leslie~Pack Kaelbling, Michael~L Littman, and Andrew~W Moore.
\newblock Reinforcement learning: A survey.
\newblock {\em Journal of artificial intelligence research}, 4:237--285, 1996.

\bibitem{Kallenberg}
Lodewijk Kallenberg.
\newblock {\em Markov Decision Processes}.
\newblock Lecture Notes. University of Leiden, 2016.

\bibitem{kemeny1969finite}
John~G Kemeny and J~Laurie Snell.
\newblock {\em Finite markov chains}, volume~26.
\newblock van Nostrand Princeton, NJ, 1969.

\bibitem{Kulkarni2011}
V.~G. Kulkarni.
\newblock {\em Discrete-Time Markov Models}, pages 5--58.
\newblock Springer New York, New York, NY, 2011.
\newblock \href {https://doi.org/10.1007/978-1-4419-1772-0_2}
  {\path{doi:10.1007/978-1-4419-1772-0_2}}.

\bibitem{lang1978elliptic}
Serge Lang.
\newblock {\em Elliptic curves: Diophantine analysis}, volume 231 of {\em
  Grundlehren der mathematischen Wissenschaften}.
\newblock Springer, 1978.

\bibitem{DBLP:conf/tacas/MertensKQW24}
Hannah Mertens, Joost{-}Pieter Katoen, Tim Quatmann, and Tobias Winkler.
\newblock Accurately computing expected visiting times and stationary
  distributions in {M}arkov chains.
\newblock In Bernd Finkbeiner and Laura Kov{\'{a}}cs, editors, {\em Tools and
  Algorithms for the Construction and Analysis of Systems - 30th International
  Conference, {TACAS} 2024, Held as Part of the European Joint Conferences on
  Theory and Practice of Software, {ETAPS} 2024, Luxembourg City, Luxembourg,
  April 6-11, 2024, Proceedings, Part {II}}, volume 14571 of {\em Lecture Notes
  in Computer Science}, pages 237--257. Springer, 2024.
\newblock \href {https://doi.org/10.1007/978-3-031-57249-4\_12}
  {\path{doi:10.1007/978-3-031-57249-4\_12}}.

\bibitem{DBLP:journals/eor/NgBCK10}
C.~T. Ng, M.~S. Barketau, T.~C.~Edwin Cheng, and Mikhail~Y. Kovalyov.
\newblock "product partition" and related problems of scheduling and systems
  reliability: Computational complexity and approximation.
\newblock {\em Eur. J. Oper. Res.}, 207(2):601--604, 2010.
\newblock \href {https://doi.org/10.1016/J.EJOR.2010.05.034}
  {\path{doi:10.1016/J.EJOR.2010.05.034}}.

\bibitem{paulsson2004summing}
Johan Paulsson.
\newblock Summing up the noise in gene networks.
\newblock {\em Nature}, 427(6973):415--418, 2004.

\bibitem{DBLP:books/wi/Puterman94}
Martin~L. Puterman.
\newblock {\em {M}arkov Decision Processes: Discrete Stochastic Dynamic
  Programming}.
\newblock Wiley Series in Probability and Statistics. Wiley, 1994.

\bibitem{DBLP:journals/mor/Rothblum84}
Uriel~G. Rothblum.
\newblock Multiplicative markov decision chains.
\newblock {\em Math. Oper. Res.}, 9(1):6--24, 1984.
\newblock URL: \url{https://doi.org/10.1287/moor.9.1.6}, \href
  {https://doi.org/10.1287/MOOR.9.1.6} {\path{doi:10.1287/MOOR.9.1.6}}.

\end{thebibliography}

\newpage
\appendix

\subsection{Comparison of succinctly represented integers (CSRI)}
\label{app:CSRI}

\begin{proposition}
\label{prop:CSRI}
Comparing succinctly represented integers (CSRI) is in $\exists \mathbb{R}$.
\end{proposition}

\begin{proof}
Given $a_1,\dots, a_n, b_1, \dots, b_n, c_1,\dots, c_m, d_1,\dots, d_m\in \mathbb{N}$, 
we have to decide whether 
\[
a_1^{b_1}\cdot \dots \cdot a_n^{b_n} > c_1^{d_1}\cdot \dots \cdot c_m^{d_m}.\tag{$\ast$}
\]
We use a binary representation to represent $b_i$ for $1\leq i \leq n$ and $d_i$ for $1\leq i \leq m$:
Let $k_{b_i} = \lceil \log b_i \rceil$ and $k_{d_i} = \lceil \log d_i \rceil$ for all $i$.
Then, 
\begin{equation*}
b_i = {\sum}_{j=0}^{k_{b_i}}  \chi_{b_i}^j\cdot 2^j 
\end{equation*}
for some $\chi_{b_i}^j\in \{0,1\}$ for all $i$ and all $j\leq k_{b_i}$. The values  $\chi_{d_i}^j\in \{0,1\}$ for all $i$ and all $j\leq k_{b_i}$ are chosen analogously.
Now, for all $1\leq i \leq n$, define the formulas
\begin{align*}
	& \phi_{a_i,b_i}(x) =  \exists x_0 \dots \exists x_{k_{b_i}} \\
	& \left(  x_0 =   a_i \land  {\bigwedge}_{j=1}^{k_{b_i}} x_j = x_{j-1}\cdot x_{j-1} \land   x = {\prod}_{j:\chi_{b_i}^j= 1} x_j \right).
\end{align*}
The subformula $x_0 =   a_i $ can be spelled out using the binary representation of $a_i$.
The formula $\phi_{a_i,b_i}(x)$ now defines $a_i^{b_i}$.

Analogously, we define the formulas $\phi_{c_i,d_i}(x)$ for $1\leq i \leq m$.
Now, inequality ($\ast$) holds if and only if the following sentence $\phi$ belongs to the existential theory of the reals:
\begin{align*}
	& \phi = \exists z_1 \dots \exists z_n \exists y_1 \dots \exists y_m                                                                                    \\
	& \left( {\bigwedge}_{i=0}^n \phi_{a_i,b_i}(z_i)\land {\bigwedge}_{i=0}^m \phi_{c_i,d_i}(y_i)) \land {\prod}_{i=1}^n z_i > {\prod}_{i=1}^m y_i \right).
\end{align*}
The length of this formula $\phi$ is polynomial in the size of the binary representation of the input and can be constructed in polynomial time.
\end{proof}

\subsection{Omitted Proofs for Markov Chains} \label{appendix:proofs}

\lemmeanpayoffMC*

\begin {proof}
	For the first part, we know that for almost all paths $\infinitepath = \infinitepath_0 \infinitepath_1 \cdots$ 
	\begin{equation*}
		\MPsup_{\log(\reward)}(\infinitepath) = \ExpectationMC<\MC, \initialstate>[\MP_{\log(\reward)}] < 0.
	\end{equation*}
	So, by definition, there is an $\varepsilon>0$ and $N \in \mathbb{N}$ such that
	\begin{equation*}
		\frac{1}{n+1} {\sum}_{i=0}^n \log(\reward(\rho_i))  < - \varepsilon
	\end{equation*}
	for all $n>N$.
	Consequently, for all $n>N$, we have
	\begin{equation*}
	\log\left({\prod}_{i=0}^n \reward(\rho_i)\right) < - \varepsilon \cdot (n+1) 
		\end{equation*}
and hence
	 \[
	  {\prod}_{i=0}^n \reward(\rho_i) < e^{- \varepsilon \cdot (n+1)}.
	  \]
	We conclude that
	\begin{equation*}
		\pathreward(\infinitepath) \leq {\limsup}_{n \to \infty} e^{- \varepsilon \cdot (n+1)} = 0.
	\end{equation*}
	So, for almost all paths $\infinitepath$, we have $\pathreward (\infinitepath) = 0$ and the claim follows.
	
	The second part follows analogously.
\end {proof}

\zeromeanpayoff*

\begin {proof}
Point~1):
If all cycles have multiplicative reward $1$, then any prefix of an infinite path $\infinitepath$
has the same multiplicative reward as some simple path (i.e.\ a path without loops).
As there are only finitely many simple paths, there are only finitely many values the multiplicative reward $\multreward_n(\infinitepath)$ after $n$ steps can take.
In particular, this set of values is independent of $n$.
Consequently, the same holds for $\pathreward(\infinitepath) = \limsup_{n\to \infty} \multreward_n(\infinitepath)$
and $\pathrewardinf(\infinitepath) = \liminf_{n\to \infty} \multreward_n (\infinitepath)$.
As these finitely many values are strictly between $0$ and $\infty$, the expected values also lie strictly between $0$ and $\infty$.

\medspace
Point~2):
A cycle $s_0 s_1 \dots s_{k+1}$ with $s_0 = s_{k+1}$ has $\prod_{i=0}^k \reward(s_i) \neq 1$ exactly when $\sum_{i=0}^k \log(\reward(s_i)) \neq 0$.
From \cite[Lem.~3.2]{BBDGS2018} (applied to the reward function $\log(\reward)$) we get that if $\ExpectationMC<\MC, \initialstate>[\MP_{\log(\reward)}] = 0$ and there exists a cycle with additive reward w.r.t.\ $\log(r)$ different to $0$, then $\liminf_{n\to\infty} \sum_{i=0}^{n} \log(\reward(\infinitepath_i)) = -\infty$ and hence $\pathrewardinf(\infinitepath)=0$ almost surely.
Analogously, $\limsup_{n\to\infty} \sum_{i=0}^{n} \log(\reward(\infinitepath_i)) = \infty$ and hence $\pathreward(\infinitepath)=\infty$ almost surely.
\end {proof}

\uniquereward*

\begin {proof}
First, assume that all cycles are 1-cycles.
Let $s, t \in \States$ two states, $\finitepath, \finitepath'$ two paths from $s$ to $t$, and $\chi$ a path from $t$ to $s$.
Now, the cycles $\finitepath\circ \chi$ and $\finitepath' \circ \chi$ both have multiplicative reward $1$. Consequently, $\multreward(\finitepath)=\multreward(\finitepath')$.

For the other direction, let $\finitepath$ an n-cycle from $s$ to $s$.
Then, $\finitepath' = \finitepath \circ \finitepath$ (i.e.\ the path $\finitepath$ repeated twice) is a path from $s$ to $s$ with $\multreward(\finitepath') = \multreward(\finitepath)^2 \neq \multreward(\finitepath)$.
\end {proof}

\finitevaluesMC*

\begin{proof}
	It is sufficient to consider only paths which visit every state infinitely often as
	this happens almost surely  in a strongly connected Markov chain.
	Let $\rho$ be any such path.
	We can partition $\rho$ into finite segments starting at $\initialstate$, visiting all states (some possibly multiple times), and coming back to $\initialstate$.
	Each such segment is a cycle and by assumption has reward $1$.
	As such, the multiplicative reward of $\rho$ until any index $j$ equals the reward obtained within the current segment.
	Within any segment, the maximal and minimal multiplicative reward naturally equals $\max_{t\in \States} R(\initialstate,t)$ and $\min_{t\in \States} R(\initialstate,t)$, respectively.
	Thus, $\pathreward(\rho) = \max_{t\in \States} R(\initialstate,t)$ and $\pathrewardinf (\rho)= \min_{t\in \States} R(\initialstate,t)$.
	As this holds for almost all paths, we are done.
\end{proof}

\prophardnessMP*

\begin {proof}
Given $a_1,\dots, a_n, b_1, \dots, b_n, c_1,\dots, c_m, d_1,\dots, d_m\in \mathbb{N}$, we construct the following strongly connected Markov chain $\MC$:
We use states $s$, $t_1,\dots, t_n$ and $q_1, \dots, q_m$ and denote the state space by $S$. Letting $\Delta= \sum_{i=1}^n b_i + \sum_{j=1}^m d_j$,  the transition probabilities are given by
\begin{align*}
	\mctransitions(s,t_i) & = \frac{b_i}{\Delta} \quad \text{ and } & \mctransitions(t_i,s) & =1  & \text{ for $1\leq i \leq n$}, \\
	\mctransitions(s,q_j) & = \frac{d_j}{\Delta} \quad \text{ and } & \mctransitions(q_j,s) & =1  & \text{ for $1\leq j \leq m$}.
\end{align*}
The rewards are given by
\begin{align*}
&\reward(s)=1, \quad \reward(t_i)=a_i \text{ for $1\leq i \leq n$}, \\
 &\reward(q_j) = \frac{1}{c_j} \text{ for $1\leq j \leq m$}.
\end{align*}
Now, the stationary distribution $\theta$ is given by
\begin{align*}
&\theta(s) = \frac12, \quad \theta(t_i) = \frac{b_i}{2\Delta}  \text{ for $1\leq i \leq n$}, \\
 &\theta(q_j) = \frac{d_j}{2\Delta}   \text{ for $1\leq j \leq m$}.
\end{align*}
So, as we have seen before, $\ExpectationMC<\MC, \initialstate>[\MP_{\log(\reward)}] > 0$ if and only if
\begin{equation*}
{\prod}_{s\in S} \reward(s)^{\theta(s)} = {\prod}_{i=1}^n a_i^{b_i / 2\Delta} \cdot  {\prod}_{j=1}^m \left(\frac1{c_j}\right)^{d_j / 2\Delta} >1.
\end{equation*}
This holds if and only if $ \prod_{i=1}^n a_i^{b_i} >  \prod_{j=1}^m c_j^{d_j}$.
\end {proof}

\lemzeros*

\begin {proof}
	We know that $\ProbabilityMC<\MC, s>[\reach T] = 1$.
	Consequently, $\ProbabilityMC<\MC, s>[\reach T_0] < 1$ exactly if there exists a path $\finitepath$ from $s$ to $T \setminus T_0$.
	If such a path exists, then we immediately get $v_s > 0$.
	Otherwise, almost all paths reach states which obtain a value of $0$, and thus $v_s = 0$.
\end {proof}

\lemmcnonnegsolution*

\begin {proof}
	Note that if $\ExpectationMC<\MC, \initialstate>[\pathreward] < \infty$, then for all states $s \in \States$, we have $v_s < \infty$ as we assume that all states are reachable.
	Clearly, $v_s \geq 0$. For \cref{eq:transient}, we write for a state $s\in S_?$:
	\begin{align*}
		v_s & = \ExpectationMC<\MC, \initialstate>[\pathreward] = \reward(s) \cdot {\sum}_{s' \in \States} \mctransitions(s, s') \cdot \ExpectationMC<\MC, s'>[\pathreward] \\
		    & = \reward(s) \cdot \big( {\sum}_{s' \in \States_?} \mctransitions(s, s') \cdot v_{s'} + {\sum}_{t \in T} \mctransitions(s,t) \cdot v_t + {}                   \\
		    & \qquad\qquad {\sum}_{s' \in \States_0 \setminus T} \mctransitions(s,s') \cdot v_{s'} \big)                                                                    \\
		    & = \reward(s) \cdot \big( {\sum}_{s' \in \States_?} \mctransitions(s,s') \cdot v_{s'} + {\sum}_{t \in T} \mctransitions(s,t) \cdot v_t \big).  
	\end{align*}
\end {proof}

\lemmcvalueiteration*

\begin {proof}
	By induction on $n$, we get for all states $s\in S_?$ that
	\begin{equation*}
		(U^n(0))_s = \int_{\reach^{\leq n} T} \pathreward \,d\ProbabilityMC<\MC, s>.
	\end{equation*}
	Observe that $\reach^{\leq n} T$ is the disjoint union of $\reach^{= n} T$ and $\Union_{n = 1}^\infty \reach^{= n} T = \reach T$.
	Moreover, $\ProbabilityMC<\MC, s>[\reach T] = 1$.
	Thus
	\begin{align*}
		\ExpectationMC<\MC, s>[\pathreward] & = \int \pathreward \,d\ProbabilityMC<\MC, s> = \int_{\reach T} \pathreward \,d\ProbabilityMC<\MC, s> \\
			& = {\sum}_{k = 0}^{\infty} \int_{\reach^{= k} T} \pathreward \,d\ProbabilityMC<\MC, s> \\
			& = {\lim}_{n\to\infty} {\sum}_{k = 0}^{n} \int_{\reach^{= k} T} \pathreward \,d\ProbabilityMC<\MC, s> \\
			& = {\lim}_{n\to\infty} \int_{\reach^{\leq n} T} \pathreward \,d\ProbabilityMC<\MC, s> = {\lim}_{n\to\infty} (U^n(0))_s.
	\end{align*}
	Note that $\pathreward(\rho) \geq 0$, thus the sum is monotonically increasing and the limits are well-defined.
\end {proof}

\cormcleast*

\begin {proof}
	As $U$ is monotonous w.r.t.\ to the component-wise ordering, and thus Scott-continuous, \cref{thm:valueiteration} implies by the Kleene fixed-point theorem that $(v_s)_{s\in S_?}$ is the least fixed point of $U$.
	However, if $x^\ast \geq 0$ satisfies \cref{eq:transient}, it is a fixed point of $U$.
\end {proof}

\cormcunique*

\begin {proof}
	As all states in $\States_?$ reach a state in $T \setminus T_0$ with positive probability, we have $v_s > 0$ for all $s \in \States_?$.
	(In other words, $\States_? \intersection \States_0 = \emptyset$.)
	Let $v$ be the vector $(v_s)_{s \in \States_?}$.
	Suppose there is a solution $y^\ast \neq v$ to \cref{eq:transient}.
	Let $\Delta = y^\ast - v \neq 0$.
	As \cref{eq:transient} is a linear equation, for any $t \in \Reals$, $v + t\cdot \Delta$ is a solution, too.
	As $\Delta \neq 0$, and $v$ is component-wise strictly positive, there is a $t$ such that $v + t \cdot \Delta \geq 0$, but $v + t \cdot \Delta \ngeq v$.
	So, $v + t \cdot \Delta$ is non-negative and satisfies \cref{eq:transient}.
	Hence, by \Cref{cor:least}, $v \leq v + t \cdot \Delta$, yielding a contradiction.
	Thus, $v$ is the unique solution to \cref{eq:transient}.
\end {proof}

\subsection{Omitted Proofs for MDP} \label{appendix:proofsmdp}

\mpMDP*

\begin {proof}
	Let $\vartheta = \ExpectationMDP<\MDP, \initialstate><\max>[\MP_{\log(\reward)}]$ the optimal mean payoff.
	
	Point~1):
	It is well-known that $\vartheta < 0$ in a strongly connected MDP implies that $\ProbabilityMDP<\MDP, s><\sched>[\MP_{\log(\reward)}\leq \vartheta] = 1$ for any state $s$ and scheduler $\sched$ (see, e.g., \cite[Sec.~6.1]{Kallenberg}).
	Using this fact, the result follows analogously to the proof of \cref{lem:meanpayoffMC}.

	Point~2):
	Similarly, in strongly connected MDPs, if $\vartheta > 0$, there exists an MD-scheduler $\sched$ with $\ProbabilityMDP<\MDP, s><\sched>[\MP_{\log(\reward)} = \vartheta] = 1$.
	Again, the claim follows analogous to \cref{lem:meanpayoffMC}.
\end {proof}

\checklogMP*

\begin {proof}
As there is an MD-scheduler maximizing the expected mean payoff, we can go through all MD-schedulers $\sched$ in polynomial space and check 
whether $\ExpectationMDP<\MDP, \initialstate><\sched>[\MP_{\log(\reward)}] \bowtie 0$ in the resulting Markov chain. Each of the checks belongs to the complexity class $\textrm{CSRI}\subseteq \textrm{PSPACE}$ (\cref{prop:logMP_MC}).
\end {proof}

\Amax*

\begin{proof}
During the check how the maximal expected logarithmic mean payoff compares to $0$ as in the proof of \cref{prop:check_logMP0}, all actions that are part of a BSCC induced by an MD-scheduler with logarithmic mean payoff $0$ can be stored.
\end{proof}

\gamblinglimsup*

\begin{proof}
A path $\infinitepath = s_1 s_2 \dots$ in a gambling max-BSCC almost surely satisfies
$\limsup_{n\to \infty} {\sum}_{i=1}^n \log(r(s_i))  = \infty$ \cite{BBDGS2018}.
Hence, almost surely $\pathreward(\infinitepath) = \infty$.
An MD-scheduler achieving value $\infty$ from all states simply realizes the gambling max-BSCC and takes actions almost surely leading to this gambling max-BSCC from all other states.
\end{proof}

\NPhardgambling*

\begin{proof}
	Containment: Given an MD scheduler, we can apply the techniques derived for Markov chains.

	Hardness: We consider the \emph{subset product} problem.
	There, we are given integers $\{n_1, \dots, n_k\}$ and $T$ (in binary), and asked to decide whether there exists a subset $I \subseteq \{1, \dots, k\}$ such that $\prod_{i \in I} n_i = T$.
	This problem is known to be NP-complete (in the strong sense) \cite{DBLP:journals/eor/NgBCK10}.
	Given an instance, we create a state for each $n_i$ with two choices $\text{pick}_i$ and $\text{omit}_i$, as well as states corresponding to these choices.
	The action $\text{pick}_i$ leads to state $\text{pick}_i$ with a reward of $n_i$, analogous for $\text{omit}_i$ with reward $1$.
	From both states, there is a single action, leading to the state corresponding to $n_{i+1}$.
	For $n_k$, both choice-states lead to a final state with reward $1/T$ which then loops back to the state for $n_1$.
	Observe that every distribution in this MDP is Dirac, and thus every MD-scheduler induces a unique cycle, whose total reward equals the product of all $\text{pick}$-ed numbers divided by $T$.
	As such, every MD-scheduler directly corresponds to a possible solution of subset-product and a 1-cycle exists exactly when the subset-product problem has a positive answer.
\end{proof}

\rewardtounique*

\begin{proof}
	Assume for contradiction that the property is violated for $s$ and $t$.
	Let $\strategy$ be an MD-scheduler which from $t$ reaches $s$ almost surely.
	Such a scheduler exists since $\MDP$ is strongly connected.
	Then, consider the schedulers $\sigma'$ and $\tau'$ which are obtained by following $\sigma$ or $\tau$ until $t$ is reached and then $\strategy$ until $s$ is reached again, switching back to $\sigma$ and $\tau$, and so on.
	We consider the \enquote{round trip} reward $\logtotalrewardto{t ; s}$, i.e.\ the reward until $t$ and after that $s$ is reached, defined analogous to $\logtotalrewardto{t}$.
	By linearity of expectation and the definition of $\sigma'$, we get that $\ExpectationMDP<\MDP, s><\sigma'>[\logtotalrewardto{t ; s}] = \ExpectationMDP<\MDP, s><\sigma>[\logtotalrewardto{t}] + \ExpectationMDP<\MDP, t><\strategy>[\logtotalrewardto{s}]$ (observe that different strategies are considered).
	By assumption, we thus get $\ExpectationMDP<\MDP, s><\sigma'>[\logtotalrewardto{t ; s}] \neq \ExpectationMDP<\MDP, s><\tau'>[\logtotalrewardto{t ; s}]$.
	However, observe that both of these values necessarily need to equal zero, since otherwise the mean payoff of $\tau'$ or $\sigma'$ w.r.t.\ $\log(\reward)$ would not be 0:
	Assume that $\ExpectationMDP<\MDP', s><\sigma'>[\logtotalrewardto{t ; s}] = c \neq 0$, i.e.\ the expected logarithmic total reward over paths from $s$ to $t$ and back to $s$ is not equal to $0$.
	By assumption, almost all paths go from $s$ to $t$ and back to $s$ infinitely often.
	Thus, the mean payoff is not equal to $0$, contradicting our assumption.
\end{proof}

\idetifygoodactions*

\begin{proof}
	Observe that for any pair of states $s$ and $t$, we have $Q(s, t) = \sum_{u \in \States} \mdptransitions(s, a, u) \cdot (\log(\reward(u)) + Q(u, t))$ for any action $a \in \Actions(s)$, as the value $Q(s, t)$ is unique.
	For a state $s$ and action $a \in \Actions^\ast(s)$, we thus have $Q(s, x) = \log(\reward(u)) + Q(u, x)$ and, equivalently, $\exp(Q(s, x)) = \reward(u) \cdot \exp(Q(u, x))$ for any successor $u$ of $s, a$ by definition of $\Actions^\ast$.

	\underline{If:}
	We show that when restricted to actions in $\Actions^\ast$, for any state $s$ in $E$ \emph{every path} from $s$ to $x$ has the same \emph{multiplicative} reward, namely $\exp(Q(s, x))$.
	We prove the statement for all paths by induction over their length.
	For $n = 0$, the statement trivially holds.
	Thus, fix an arbitrary $n$ for which the statement holds.
	Let $\rho$ and $\rho'$ two paths of length between $1$ and $n + 1$, both going from state a $s$ to $x$.
	Let $u$ and $u'$ the second states of these paths, respectively.
	By the induction hypothesis, the multiplicative rewards of the remaining paths from $u$ and $u'$ to $x$ equal $\exp(Q(u, x))$ and $\exp(Q(u', x))$.
	Consequently, the multiplicative rewards of $\rho$ and $\rho'$ equal $\reward(u) \cdot \exp(Q(u, x))$ and $\reward(u') \cdot \exp(Q(u', x))$, respectively.
	Since $u$ and $u'$ are successors of $s$ using actions from $\Actions^\ast$, we have that $\exp(R(s, x)) = \reward(u) \cdot \exp(Q(u, x)) = \reward(u') \cdot \exp(Q(u', x))$.
	Consequently, the rewards of $\rho$ and $\rho'$ both equal $\exp(Q(s, x))$.

	Now, we extend this statement to any pair of states, i.e.\ every path from a state $s$ to a state $t$ using only actions of $\Actions^\ast$ has reward $\exp(Q(s, t))$.
	For contradiction, assume there are two paths $\rho$ and $\rho'$ from $s$ to $t$ which obtain a different reward.
	Then, let $\rho''$ any path from $t$ to $x$ in $E$ (which exists, as $E$ is an EC).
	Concatenating $\rho$ and $\rho'$ with $\rho''$ yields two paths from $s$ to $x$ in $E$ with different reward, contradicting the above.

	Together, we have that any path between two states $s$ and $t$ in $E$ yields the reward $\exp(Q(s, t))$.
	In particular, recall that $\exp(Q(x, x)) = 1$, thus any cycle only using actions from $\Actions^\ast$ yields a reward of $1$.

	\underline{Only if:}
	Fix an EC $E$ and states $s$ and $x$ in $E$.
	Moreover, assume that in $s$ an action $a \notin \Actions^\ast(s)$ is used.
	Let $u$ and $u'$ two successors of $(s, a)$ with $\log(\reward(u)) + Q(u, x) \neq \log(\reward(u')) + Q(u', x)$ (which exist as $a \notin \Actions^\ast(s)$).
	W.l.o.g.\ choose $u$ such that this value is minimal among all successors, and $u'$ such that it is maximal.
	Recall that $Q(s, x) = \sum_{t \in \States} \mdptransitions(s, a, t) \cdot (\log(\reward(t)) + Q(t, x))$ and $Q(t, x)$ is the expected additive reward from $t$ to $x$ w.r.t.\ $\log(\reward)$ (under any strategy reaching $x$ from $t$ almost surely). %
	Consequently, $\log(\reward(u)) + Q(u, x) < Q(s, x) < \log(\reward(u')) + Q(u', x)$ by choice of $u$ and $u'$; and there must exist paths $\rho$ and $\rho'$ which obtain an additive reward of less and more than $Q(s, x)$, respectively.
	In particular, their additive (and thus also multiplicative) reward is distinct.
	As both $x$ and $s$ are in $E$, there also exists a path $\rho''$ from $x$ to $s$.
	Clearly, concatenating each $\rho$ and $\rho'$ with $\rho''$ yields two cycles through $s$ which yield two different (additive and) multiplicative rewards.
\end{proof}

\computeAstar*

\begin{proof}
	First, fix an arbitrary state $x$ and MD-scheduler $\sched$ under which $x$ is reached with probability 1 from every state.
	Since the MDP is strongly connected, such a scheduler always exists.
	Consider the Markov chain induced by the scheduler $\sched$.
	In this MC, make $x$ absorbing and set its reward to $0$.
	Clearly, the expected reward until reaching $x$ corresponds to the expected total reward in that Markov chain.
	For each starting state $s$, compute the \emph{expected visiting time} $\text{evt}_s(u)$ (see \cite{DBLP:conf/tacas/MertensKQW24} for a recent treatment of EVTs), i.e.\ the number of times state $u$ is visited on average when moving from $s$ to $x$ under $\sched$.
	This can be done in polynomial time by solving a linear equation system \cite[Cor.~3.3.6]{kemeny1969finite} (rephrased in \cite[Thm.~1]{DBLP:conf/tacas/MertensKQW24}).
	By \cref{lem:reward_to_is_unique}, we have that $Q(s, x) = \ExpectationMDP<\MDP, s><\sched>[\logtotalrewardto{x}]$.
	Now, observe that the entire expected reward can be expressed as the sum of rewards where rewards are only obtained in one fixed state $s$ and $0$ otherwise.
	Then, as a consequence of \cite[Lem.~1]{DBLP:conf/tacas/MertensKQW24}, we get that expected total reward can be expressed through EVT as $\ExpectationMDP<\MDP, s><\sched>[\logtotalrewardto{x}] = \sum_{u \in \States} \text{evt}_s(u) \cdot \log(\reward(t))$.\footnote{Note the similarity to the treatment of $\ExpectationMC<\MC, s>[\MP_{\log(\reward)}]$ in the previous section: we use the expected visiting times to derive the total reward instead of the stationary distribution to derive the mean payoff.
	See also \cite{DBLP:conf/tacas/MertensKQW24} for a connection between stationary distribution and EVT.}
	This value is, in general, irrational.
	However, to compute whether $a \in A^\ast(s)$ for a state-action pair $(s, a)$, we do not need to compute these values precisely, but only decide whether all successors have the same value (as seen in \cref{stm:identify_good_actions}).
	Clearly, if an action has only one successor, the condition is trivially satisfied, and $a \in A^\ast(s)$.
	Otherwise, fix an arbitrary successor $t$ and for every other successor $t'$ we check whether $\log(\reward(t)) + Q(t, x) = \log(\reward(t')) + Q(t', x)$.
	We again get a query of the form \enquote{$\sum p_i \cdot \log(r_i) = 0$}, and the same techniques as in \cref{prop:logMP_MC} can be applied to reduce this to an ESRI instance, which is solvable in polynomial time.
	We only need to perform linearly many such comparisons, so, in total, we obtain the result.
\end{proof}

\correctnessspider*

\begin {proof}
	We can mimic any scheduler $\strategy$ for $\MDP$ by a scheduler $\strategy_E$ for $\MDP_E$ and vice versa:
	Entering $E$ and staying there corresponds to moving to state $\bot$ via action $\mathit{stay}$ in $\MDP_E$.
	The correctness of the values $\overline{v_E}$  or  $\underline{v_E}$  used for action $\mathit{stay}$, depending on whether we investigate $\pathreward$ or $\pathrewardinf$ can be seen as follows:
	
	Suppose a scheduler visits $E$ at a state $t$ and stays in $E$ forever with some positive probability. Then almost all of the paths  $\infinitepath$ staying in $E$ will visit all states of the BSCC $E$ 
	infinitely often. As all paths from $t$ to another state $t'$ have multiplicative reward $R(t,t')$, the prefixes of $\infinitepath$ starting from $t$ will almost surely have multiplicative rewards
	$R(t,t')$ for all $t'$ in $E$ infinitely often. Hence, $\pathreward(\infinitepath)$ starting from $t$ will almost surely be $  \max_{t'\in T} R(t,t') = R(t,s_E) \cdot \max_{t'\in T} R(s_E,t')$.
	But this is exactly the multiplicative reward of the path 
	from $t$ to $s_E$ and then to $\bot$ via actions $\mathit{center}_t$ with reward $R(t,s_E)/\reward(t)$ where $\reward(t)$ cancels out when it is obtained by leaving $t$ and $\mathit{stay}$ with reward $\overline{v_E}=\max_{t'\in T} R(s_E,t')/\reward(s_E)$ where $\reward(s_E)$ cancels out when it is obtained while leaving $s_E$.
	Similarly, $\pathrewardinf(\infinitepath)$ almost surely is $\min_{t'\in T} R(t,t') = R(t,s_E) \cdot \min_{t'\in T} R(s_E,t')$ and the analogous reasoning applies.
	
	Leaving $E$ via state action pair $(t, \beta)$ on the other hand corresponds to moving to $s_E$ and taking action $\tau_{t,\beta}$ there.
	The correctness of the reward along the corresponding path follows as above. Using that $R(t,s_E)\cdot R(s_E,t)=1$.
\end {proof}

\spidernumber*

\begin{proof}
Let $E=(T,B)$ be the removed non-gambling max-BSCC and let $s_E$ be the chosen center state for the adapted spider construction.
All auxiliary states that are added in the adapted spider construction have only one enabled action and hence do not influence the sum in the statement of the lemma.
For all states $t$ in the BSCC $E$, all actions $\beta \not=B(t)$ are removed and instead an action $\tau_{t,\beta}$ is enabled in $s_E$. 
For states $t\not=s_E$, furthermore, the action $B(t)$ is replaced by $\mathit{center}_t$. Finally, in state $s_E$, the action $B(s_E)$ is replaced by $\mathit{stay}$.
So, the number $\sum_{t\in S'} \left|\Actions'(t)\setminus\{\mathit{stay}\}\right| - 1 $ is one lower than $\sum_{t\in S} \left|\Actions(t)\right| - 1$.
\end{proof}

\boundspider*

\begin{proof}
By \Cref{lem:spider_number_iterations}, each application of the spider construction reduces the total  number of actions different to $\mathit{stay}$ enabled in states belonging to the original state space $\States$ by one.
Our iterative procedure never chooses new auxiliary states as center of the spider construction and the action $\mathit{stay}$ cannot be part of a BSCC. So, after at most $k=\sum_{t\in S} \left|\Actions(t)\right|  $
applications of the spider construction, there are no non-gambling max-BSCCs left as after $k$ applications of the spider construction $\mathit{stay}$ would be the only action enabled in states from the original state space $\States$.
\end{proof}

\ECliminf*

\begin{proof}
If no absorbing state is reachable form $s$, under any scheduler $\sched$, we have that almost all paths $\infinitepath$ either achieve a logarithmic mean payoff $<0$ or they achieve a logarithmic mean payoff of $0$ and take cycles with multiplicative reward different to $1$
infinitely often. In the latter case, with probability $1$, the multiplicative reward of prefixes of $\infinitepath$ comes close to $0$ infinitely often. So, almost all paths $\infinitepath$ satisfy  $\pathrewardinf(\infinitepath)=0$.

If an absorbing state $u$ with value $1$ is reachable from $s$, consider the following scheduler $\sched_K$ for $K\in \Reals$.
The scheduler stays in the end component $\cE$ by realizing a BSCC with
logarithmic mean payoff $0$ until the multiplicative reward exceeds $K$. This will happen with probability $1$.
Afterwards, it tries to follow a simple path $\rho$ to $u$.
It will succeed with positive probability $p$. So, the expected value of $\pathrewardinf$ under this scheduler is at least $K\cdot \reward(\rho)\cdot p$.
As this value tends to $\infty$ for $K\to \infty$, we conclude $\ExpectationMDP<\MDP, \initialstate><\max>[\pathrewardinf]=\infty$.
We remark that even one concrete scheduler can achieve infinite reward.
Intuitively, with probability $\frac{1}{2^i}$, the scheduler replicates $\sched_{4^i}$, and its expected value is $\infty$.
\end{proof}

\feasabilityLP*

\begin{proof}
We prove the contraposition of the claim for $\pathreward$. The proof works completely analogously for $\pathrewardinf$.

Let $v_s = \ExpectationMDP<\MDP, s><\max>[\pathreward]$ for all $s\in \States$.
If $\ExpectationMDP<\MDP, \initialstate><\max>[\pathreward] <\infty$, then all these values are finite as every state is reachable in $\MDP$.
The first two constraints of the LP are satisfied by the values $(v_s)_{s\in \States}$.
For the last constraint, let $s\in S\setminus T$ and $\alpha \in A(s)$ be given.
Now, a scheduler $\sched_\alpha$ that starts in $s$ by choosing $\alpha$ and afterwards following optimal schedulers starting from all successor states achieves 
$
	\ExpectationMDP<\MDP, s><\sched_\alpha>[\pathreward] = \reward(s) \cdot {\sum}_{s' \in \States} \mdptransitions(s,\alpha,s') \cdot v_{s'}
$.
Hence,
$
v_s =  \ExpectationMDP<\MDP, s><\max>[\pathreward] \geq \reward(s) \cdot {\sum}_{s' \in \States} \mdptransitions(s,\alpha,s') \cdot v_{s'}
$.
So, if the actual values are finite, the LP is feasible.
\end{proof}

\knastertarski*

\begin{proof}
By the Knaster-Tarski theorem, $T$ has a least fixed point on the complete lattice between $0$ and $x$ with respect to the componentwise ordering $\leq$.
Furthermore, this least fixed point is the least point $y$ with $T(y) \leq y$.
The constraints of the LP exactly express that $T(x)\leq x$.
As the objective function minimizes the sum of the components, the solution $x$ has to be equal to the componentwise least point $y$ with $T(y) \leq y$, hence the solution is the least fixed point of $T$.
\end{proof}

\nstepconverges*

\begin {proof}
For any scheduler $\sched$ and any state $s$, we have 
\begin{align*}
	\ExpectationMDP<\MDP, s><\sched>[\pathreward] & = \ProbabilityMDP<\MDP, s><\sched>[\lozenge T] \cdot \ExpectationMDP<\MDP, s><\sched>[\pathreward \mid \lozenge T] \tag{$\ast$}       \\
	                                              & \quad + \ProbabilityMDP<\MDP, s><\sched>[\neg \lozenge T] \cdot \ExpectationMDP<\MDP, s><\sched>[\pathreward \mid \neg \lozenge T].
\end{align*}
As all end components have negative logarithmic mean payoff, $\ExpectationMDP<\MDP, s><\sched>[\pathreward \mid \neg \lozenge T] =0$.

On the other hand,
\begin{equation*}
	\lim_{n\to \infty}  \ExpectationMDP<\MDP, s><\max>[\multreward_n] = {\sup}_{\sched} \ProbabilityMDP<\MDP, s><\sched>[\lozenge T] \cdot \ExpectationMDP<\MDP, s><\sched>[\pathreward \mid \lozenge T].
\end{equation*}
Then, by ($\ast$),
\begin{equation*}
	{\sup}_{\sched} \ProbabilityMDP<\MDP, s><\sched>[\lozenge T] \cdot \ExpectationMDP<\MDP, s><\sched>[\pathreward \mid \lozenge T] = \ExpectationMDP<\MDP, s><\max>[\pathreward].\qedhere
\end{equation*}
\end {proof}

\schedulercomplexity*

\begin{proof}
If $\ExpectationMDP<\MDP, \initialstate><\max>[\pathreward]<\infty$, we can extract an MD-scheduler from the solution $x$ of the linear program: For all states with value $0$ in the solution, the chosen action does not matter.
(Since the values of the LP are correct, $\ExpectationMDP<\MDP, \initialstate><\sched>[\pathreward] \leq \ExpectationMDP<\MDP, \initialstate><\max>[\pathreward]$ and we cannot obtain a value better than $0$.)
For all remaining states $s\in \States\setminus T$, we choose an action $\alpha\in \Actions(s)$ such that equality holds in the corresponding constraint, i.e.\ such that $x_s =  \reward(s) \cdot {\sum}_{s' \in \States} \mdptransitions(s,\alpha,s') \cdot x_{s'}$.
As such, the values $x_s$ for states with value other than $0$ form a non-negative solution to \cref{eq:transient} in \cref{sec:MC} for the Markov chain induced by $\sched$.
In other words, the resulting scheduler achieves value $x_{\initialstate}$.
Further, the use of actions $\mathit{stay}$ and $\mathit{center}$ as well as $\tau_{t,\beta}$ introduced by the adapted spider construction can be mapped back to MD-behavior in the original MDP:
If $\mathit{stay}$, the MD-scheduler realizing the respective BSCC stays in the BSCC forever.
If $\mathit{center}$ as well as $\tau_{t,\beta}$ are used, an MD-scheduler can traverse the BSCC until $t$ is reached and then use action $\beta$ there.
This can be applied iteratively to all BSCCs to which the spider construction had been applied (reversing the order in which the construction was applied to the BSCCs).

If $\ExpectationMDP<\MDP, \initialstate><\max>[\pathreward] = \infty$ due to a gambling max-BSCC or BSCC with positive expected logarithmic mean payoff, an MD-scheduler can realize this BSCC as shown in \Cref{lem:gambling_limsup}.

Finally, it is possible that $\ExpectationMDP<\MDP, \initialstate><\max>[\pathreward]=\infty$  due to the transient behavior, i.e.\ not due to some gambling max-BSCC or BSCC with positive expected logarithmic mean payoff.
So, we can assume \Cref{ass:negativeMP} and conclude that \ref{eq:mdp_lp} is not feasible in this case.
Suppose now towards a contradiction that every MD-scheduler $\sched$ satisfies $\ExpectationMDP<\MDP, \initialstate><\sched>[\pathreward]<\infty$. 
For each scheduler $\sched$, let $x^\sched = (x^\sched_s)_{s\in \States} = (\ExpectationMDP<\MDP, s><{\sched}>[\pathreward])_{s\in \States}$ be the vector of expected values from each state. 
Now, pick a scheduler $\tsched$ that always chooses locally optimal actions $\alpha\in \Actions(s)$ in all states $s\in \States\setminus T$, i.e.\ actions with
\begin{multline*}
	\reward(s) \cdot {\sum}_{s' \in \States} \mdptransitions(s,\alpha,s') \cdot x^{\tsched}_{s'} =                               \\
	{\max}_{\beta\in \Actions(s)}  \reward(s) \cdot {\sum}_{s' \in \States} \mdptransitions(s,\alpha,s') \cdot x^{\tsched}_{s'}.
\end{multline*}
Such a scheduler can be found by iteratively replacing locally non-optimal actions.
Now, $x^{\tsched}$ is a (finite) solution to \ref{eq:mdp_lp}, yielding a contradiction.

For $\pathrewardinf$, when $\ExpectationMDP<\MDP, \initialstate><\max>[\pathrewardinf]<\infty$ the optimal MD-scheduler can be found analogous to the above. 
\end{proof}

\thresholdproblem*

\begin{proof}
	A non-deterministic polynomial time algorithm simply guesses an MD-scheduler and checks whether the value in the resulting Markov chain exceeds the threshold.
	This is doable in polynomial time with a CSRI-oracle as shown in \Cref{sec:MC}.
\end{proof}

\end{document}